\documentclass[twoside,11pt]{article}

\usepackage{blindtext}
\usepackage[abbrvbib, preprint]{jmlr2e}
\usepackage{natbib} 
    \renewcommand{\bibsection}{\subsubsection*{References}}
\usepackage{mathtools} 
\usepackage{dsfont}
\usepackage{mathrsfs}
\usepackage{array}
\usepackage{xcolor}
\usepackage{enumitem}   
\usepackage[a4paper]{geometry}
\usepackage{amssymb, amsmath}
\usepackage{graphicx}
\usepackage{comment}
\usepackage{tabularx}
\usepackage{float}
\usepackage{algorithm}
\usepackage{algpseudocode}
\usepackage{setspace}
\usepackage{multirow}
\usepackage{booktabs}
\usepackage{enumerate}
\usepackage{bbold}
\usepackage[caption=false]{subfig} 

\newtheorem{prop}{Proposition}
\def\half{\hbox{$1\over2$}}
\newcommand{\bestmodel}[1]{\mathcal{M}_{\widehat{k}(#1)}}
\newcommand{\model}[1]{\mathcal{M}_{k(#1)}}
\newcommand{\bestparam}[1]{\widehat{\theta}_{\widehat{k}(#1)}}
\newcommand{\param}[1]{\widehat{\theta}_{k(#1)}}
\DeclareMathOperator*{\argmax}{argmax}

\ShortHeadings{}{}
\firstpageno{1}

\begin{document}

\title{A general framework for probabilistic model uncertainty}

\author{\name Vik Shirvaikar$\hspace{0.3mm}^{1}$ 
        \email vik.shirvaikar@spc.ox.ac.uk \\
        \name Stephen G. Walker$\hspace{0.3mm}^{2}$ 
        \email s.g.walker@math.utexas.edu \\ 
        \name Chris Holmes$\hspace{0.3mm}^{1}$ 
        \email chris.holmes@stats.ox.ac.uk \\}

\maketitle

\begin{abstract}
Existing approaches to model uncertainty typically either compare models using a quantitative model selection criterion or evaluate posterior model probabilities having set a prior. In this paper, we propose an alternative strategy which views missing observations as the source of model uncertainty, where the true model would be identified with the complete data. To quantify model uncertainty, it is then necessary to provide a probability distribution for the missing observations conditional on what has been observed. This can be set sequentially using one-step-ahead predictive densities, which recursively sample from the best model according to some consistent model selection criterion. Repeated predictive sampling of the missing data, to give a complete dataset and hence a best model each time, provides our measure of model uncertainty. This approach bypasses the need for subjective prior specification or integration over parameter spaces, addressing issues with standard methods such as the Bayes factor. Predictive resampling also suggests an alternative view of hypothesis testing as a decision problem based on a population statistic, where we directly index the probabilities of competing models. In addition to hypothesis testing, we demonstrate our approach on illustrations from density estimation and variable selection.
\end{abstract}

{\let\thefootnote\relax\footnote{{$^1$Department of Statistics, University of Oxford }}}
{\let\thefootnote\relax\footnote{{$^2$Department of Statistics and Data Sciences, University of Texas at Austin}}}

\section{Introduction}

Missing information is the source of statistical uncertainty. Given an observed sample from a population, any identifiable quantity of interest (such as a parameter) would be fully determined by definition if we had access to the complete population. In this paper, we expand upon this principle to provide a general construction for probabilistic model uncertainty. Our approach uses \emph{predictive resampling}, as described by \citet{fong_martingale_2024}, to recursively alternate between selecting the best current model and sampling a new observation from the model. Ultimately, this allows us to retrieve Monte Carlo uncertainty on the best model as identified with the complete missing information and a specified criterion for model selection. We motivate this approach in detail, and demonstrate its application to hypothesis testing and other decision problems. 

If $x_{1:n}$ is the observed data, then any associated uncertainty --- for example, as to which model is best --- flows from the unseen $x_{n+1:\infty}$, or in practice, $x_{n+1:N}$ for some sufficiently large $N$. With this view, uncertainty quantification requires the construction of a generative, or predictive, model for $p(x_{n+1:\infty} \mid x_{1:n})$, i.e., a joint distribution for the missing observations given what has been observed. The missing population-scale data is imputed by sampling from $p(x_{n+1:N} \mid x_{1:n})$, where taking $N \rightarrow \infty$ then provides Monte Carlo estimates and uncertainty for the quantity of interest, as a function of the complete data. This predictive approach is fundamentally different to the usual frequentist view of uncertainty, based on repeated experiments of size $n$, but also deviates from the usual Bayesian method of eliciting a prior distribution to conduct inference directly on a parameter.

\citet{fong_martingale_2024} demonstrate how the general specification of $p(x_{n+1:\infty} \mid x_{1:n})$ can be achieved through predictive resampling, where a one-step update
\begin{equation} \label{eq:update1}
\{p(\cdot \mid x_{1:i}), x_{i+1}\} \rightarrow p(\cdot \mid x_{1:i+1})
\end{equation}
recursively alternates between imputing the next data point, $x_i \sim p(\cdot\mid x_{1:i-1})$ for $i>n$, and updating the predictive distribution. Hence,
\begin{equation}
p(x_{n+1:N} \mid x_{1:n})=\prod_{i=n+1}^{N} p(x_i\mid x_{1:i-1}).   
\end{equation}
Importantly, this construction bypasses the need for a typical Bayesian likelihood-prior update, with probabilistic statements expressed directly in terms of observable data \citep{de_finetti_prevision_1937, fortini_predictive_2012}.

The above argument has been applied in the context of parameter estimation or conditional prediction \citep{holmes_statistical_2023}, but we now extend it to model uncertainty. If we consider a set of candidate models $\{\mathcal{M}_k\}_{k=1}^K$ for the data, the uncertainty in the optimal model still arises from the missing $x_{n+1:\infty}$. We require a generative model $p(x_{n+1:\infty} \mid x_{1:n})$ for the missing data given the observed data, and at any given point, the natural choice for the prediction of new information is the best model available with the current information. 

To impute the missing data, we therefore revise the one-step predictive update from equation (\ref{eq:update1}) with respect to the candidate models. In particular, for $x_{1:n}$, we use a model selection criterion $C$ to determine the best current model, which we write as $\bestmodel{n}$, along with associated parameter value(s) $\bestparam{n}$. (We discuss the specification of $C$ below.) We sample a new observation from $p(\cdot \mid \bestmodel{n}, \bestparam{n})$, add it to the observed data, and recursively repeat this process up to some sufficiently large $N$. The update becomes
\begin{align} \label{eq:update2}
x_{n+1} &\sim p(\cdot \mid \bestmodel{n}, \bestparam{n}), \\
\hat{k}(n+1) &= \argmax_k C(\mathcal{M}_{k(n+1)}, x_{1:n+1}).
\end{align}
We aim to show that the model choice $\widehat{k}(n)$ converges as $n \to \infty$. 

Rather than evaluating the evidence for models under limited data, our inferential target shifts to determining different possible completions of the dataset, across which we then retrieve Monte Carlo uncertainty on the final selected model. 
This is given by
\begin{equation}
p(\mathcal{M}_k \mid x_{1:n}) = B^{-1}\sum_{b=1}^B \mathbb{1} \left(\bestmodel{\infty}^{(b)} = \mathcal{M}_k \right)
\end{equation}
as the posterior probability of each model, where $b$ indexes the repeated replications of $x_{n+1:\infty}$ over $B$ trials. (In practice, we assess the behavior of $C$ as the sample size grows and stop at some sufficiently large $N$ when the choice of model has clearly converged.) The resulting estimates express uncertainty over the space of candidate models, without requiring the elicitation of prior distributions on either the models or their parameters.

A key ingredient in this procedure is the one-step model selection criterion $C$. Since the eventual goal is to make inferences based on the complete data, we require a criterion which is consistent, or guaranteed to select the correct model as $n \rightarrow \infty$ \citep{claeskens_model_2008}. The above pipeline can then be viewed as a way of converting our consistent criterion directly into a statement of posterior uncertainty over the space of models. Current approaches to this question often rely on bootstrapping and qualitative assessment of robustness. Instead, by propagating uncertainty through the unseen data, we retrieve a probabilistic measure over which model would be best for the complete data under the specified criterion. As we will discuss, the requirement of consistency admits the use of the Bayesian information criterion (BIC), but not the Akaike information criterion (AIC) or leave-one-out cross-validation (LOO-CV). The AIC and LOO-CV are inconsistent, in that they do not necessarily select the correct model as $n \rightarrow \infty$, and so should not be used in a setting where our goal is to understand uncertainty in decisions based on the complete data.

The work of \citet{draper_assessment_1995} on model expansion provides a useful context for our approach. Draper denotes a model $\mathcal{M} = (S, \theta)$ as a set of ``structural assumptions" $S$ (such as the assumed linear structure, the link function in a GLM, etc.) along with parameter(s) $\theta$. He notes that statistical applications typically assume a best structure $S^*$ and then discuss parametric uncertainty on $\theta$, but that this equates to placing an overly concentrated prior point mass of one on $S^*$, leading to overconfident conclusions. A preferable approach would be to propagate structural uncertainty by placing a more diffuse prior $p(S)$ across a wider space of models. Draper suggests that this wider space could be determined by ``starting with a single structural choice $S^*$ and expanding it in directions suggested by context", though the specific prescription for this expansion is determined on a case-by-case basis. 

Our approach is also a form of model expansion around the initial $\bestmodel{n}$, where the progressive search for new models is guided by the imputation of unseen data. Figure \ref{fig:jellyLL} represents this idea with a schematic diagram based on a simple hypothesis test, discussed in further detail in Section \ref{sec:simplehyp}. Each individual path tracks the value of a sufficient summary statistic (in this case, the mean) for one possible realization of the complete data as new samples are imputed. All paths start at a common initial best model -- in this case, the alternative hypothesis $H_1$ -- but by recursively sampling new observations and then updating the choice of best model, we introduce uncertainty over the model space, which converges once the sample size becomes sufficiently large.

\begin{figure}
\centering
\includegraphics[width=\textwidth]{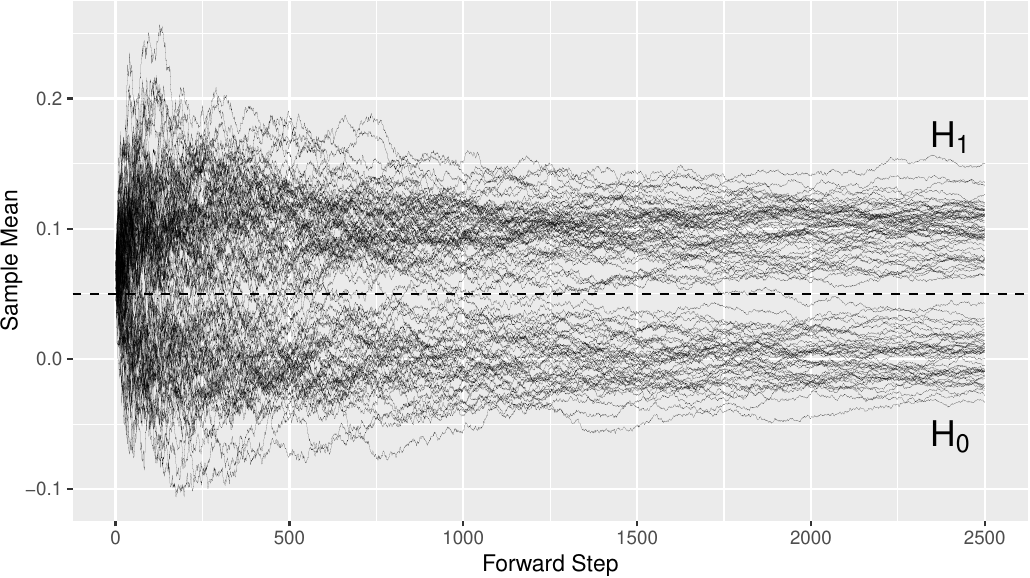}
\caption{Schematic diagram showing propagation of uncertainty through sampling of missing observations, where different possible realizations of the complete data start at the same $\bestmodel{n} = H_1$, but individually converge to either $H_0$ or $H_1$.}
\label{fig:jellyLL}
\end{figure}

The practical implications of our work are similar to those expressed in the ``prequential'' (predictive sequential) approach of \citet{dawid_present_1984}. Dawid states that ``the purpose of statistical inference is to make sequential probability forecasts for future observations rather than to express information about parameters." Our overarching motivation goes one step further - we forecast soon-to-be-observed data in order to then index uncertainty over possible model choices - but the underlying statistical object is very similar. In particular, Dawid argues for the plug-in predictive, using (for example) the maximum likelihood estimate, and specifies the joint distribution
\begin{equation}
\prod_{i=1}^n p(x_i \mid x_{i-1},\widehat\theta_{i-1})
\end{equation}
to evaluate the likelihood of a model based on observed data. We extend this plug-in strategy to select the current best model $\bestmodel{n}$ and associated parameter value $\bestparam{n}$, as determined by a (possibly dimension-penalized) likelihood function. The sequential forecasting of unseen data is valuable not only in its own right, but also in terms of what it progressively reveals about the corresponding model uncertainty.

What distinguishes this framework from traditional approaches to model uncertainty? Standard calculations using Bayes factors \citep{kass_bayes_1995} exhibit a known over-dependence on the model priors, and entail the computational complexity of integrating over the complete parameter space within a model. There have been several variations on the Bayes factor theme, such as fractional \citep{ohagan_fractional_1995} and intrinsic \citep{berger_intrinsic_1996} Bayes factors, but these require the loss of some training data or the selection of an arbitrary weight in order to calibrate an ``objective" prior. Alternative approaches apply sampling strategies to explore the space of models \citep{george_variable_1993, madigan_bayesian_1995}, of which reversible-jump Markov chain Monte Carlo (RJ-MCMC) is perhaps the most well-known; see \cite{green_reversible_1995}. MCMC-style methods require the construction of an ergodic Markov chain that can eventually visit all parts of the model space. 

In contrast, by focusing exclusively on the prediction of observable data points to propagate uncertainty, we bypass the need for any subjective prior distribution over the space of candidate models or their parameters. Predictive resampling makes complete use of the information contained in the training data, with a fully pre-specified protocol to define $p(\mathcal{M}_k \mid x_{1:n})$. Additionally, the procedure is pragmatic and straightforward. We reduce uncertainty quantification to two simple steps, model comparison and simulation of a new observation, which can be easily parallelized over the $B$ trials. The first step, which requires the optimization using a consistent model selection criterion, can leverage modern methods for efficient model search, and is computationally simpler than the transformations required to jump between dimensions in sampling methods such as RJ-MCMC.

The remainder of the paper is organized as follows. In Section 2, we formally present our framework for model uncertainty in the context of traditional approaches. Section 3 focuses specifically on hypothesis testing, and compares predictive resampling with traditional frequentist and Bayesian methods. Section 4 contains further theory related to the necessary conditions for model convergence. Section 5 contains illustrations, and Section 6 provides a discussion and conclusions. Code is available at https://github.com/vshirvaikar/MPmodel.

\section{Model uncertainty via predictive resampling}

Before detailing our approach to quantifying model uncertainty, we outline the usual Bayesian strategies.

\subsection{The standard Bayesian approach}

Consider candidate models $\{\mathcal{M}_k\}_{k=1}^K$ that aim to describe some observed data $\mathcal{D}$. Bayesian approaches to model uncertainty target the posterior probability of each individual model given the data,
\begin{equation}   
P(\mathcal{M}_k \mid \mathcal{D}) = \frac{P(\mathcal{D} \mid \mathcal{M}_k) P(\mathcal{M}_k)}{P(\mathcal{D})}
\end{equation}
where $P(\mathcal{M}_k)$ is the prior probability of model $\mathcal{M}_k$ and $P(\mathcal{D} \mid \mathcal{M}_k)$ is the marginal likelihood of the data under model $\mathcal{M}_k$. These posterior probabilities may be of interest in their own right, or as an intermediate step for model averaging,
\begin{equation}
P(\xi \mid \mathcal{D}) = \sum_{k=1}^K P(\xi \mid \mathcal{M}_k) P(\mathcal{M}_k \mid \mathcal{D}),    
\end{equation}
where $\xi$ is a separate target quantity, such as a prediction or parameter, that could vary considerably across different models \citep{leamer_specification_1978, hoeting_bayesian_1999}.

The marginal likelihood, also known as the evidence, is given by
\begin{equation}
P(\mathcal{D} \mid \mathcal{M}_k) = \int P_{\mathcal{M}_k}(\mathcal{D} \mid \theta_k) P(\theta_k) d\theta_k
\end{equation}
and integrates over the parameters that are specific to each model. This has been described as the Bayesian encoding of ``Occam's razor'' by \citet{mackay_bayesian_1992}, since simpler models place a greater probability weight on a narrower range of possible datasets, and therefore return a greater likelihood if consistent with the observed data.

The normalizing constant $P(\mathcal{D})$ cancels out, so model selection often reduces to simply comparing the marginal likelihood or evidence $P(\mathcal{D} \mid \mathcal{M}_k)$ for each model, weighted by the prior probability $P(\mathcal{M}_k)$. In the case of two competing models $\mathcal{M}_1$ and $\mathcal{M}_2$, the ratio between these quantities is referred to as the Bayes factor
\begin{equation}
BF = \frac{P(\mathcal{D} \mid \mathcal{M}_1) P(\mathcal{M}_1)}{P(\mathcal{D} \mid \mathcal{M}_2) P(\mathcal{M}_2)} 
\end{equation}
with a value of $BF>1$ presumably indicating support for $\mathcal{M}_1$ over $\mathcal{M}_2$, and vice versa. 

A well-known problem with the Bayes factor is that it depends on the model priors substantially, especially as the number of candidate models grow. As with many Bayesian analyses, a ``uniform prior'' with constant $p(\mathcal{M}_k)$ for all models will not be truly uniform in the likely event that the ``size'' of each model is different. Efforts to specify objective prior distributions have led to an array of Bayes factor alternatives, such as intrinsic and fractional Bayes factors \citep{kass_bayes_1995}. These methods use the observed data to help specify the prior in some way, such as by weighting the likelihood \citep{ohagan_fractional_1995} or setting aside a portion of data for ``training" \citep{berger_intrinsic_1996}. However, this still requires an element of user choice, and can also result in the loss of some information contained within the observed data.

Another frequently-cited problem with the Bayes factor is the computational complexity of specifying a parameter prior and integrating over the full parameter space for each model. The Bayesian information criterion (BIC) provides an approximation to the negative log-evidence, given by
\begin{equation}
\text{BIC} = d\log n-2\log \widehat{\mathcal{L}}    
\end{equation}
where $d$ is the dimension of the model, $n$ is the sample size, and $\widehat{\mathcal{L}}$ is the maximum likelihood of the model at the optimal parameter values. (The BIC is sometimes denoted as the negative of the above; in our case, a lower value indicates a better model fit.) The penalty for model complexity guards against overfitting, which could occur when directly optimizing $\log \widehat{\mathcal{L}}$; the resulting remainder term for the BIC's approximation to the log-evidence is bounded in $n$ \citep{schwarz_estimating_1978}. A tempting idea is then computing $\exp(-{\text{BIC}})$ to target the marginal likelihood directly, but \citet{kass_bayes_1995} show that this has a relative error of $O(1)$ in approximating the Bayes factor, meaning that even for large samples, the BIC should not be used to evaluate exact posterior probabilities.

\subsection{Recursive model updating}
\label{sec:mcmethod}

Here we present a predictive resampling approach that emphasizes missing data as the source of model uncertainty. With observed $x_{1:n}$, the guiding principle is that uncertainty quantification for any statistical task, including model selection, requires the construction of a model for the data we have not observed, given what has been observed. Algorithm 1 outlines the procedure, alternating between the selection of the best available model at any given point and the imputation of a new observation. This is presented for finite $N$, and the choice of model converges as $N \rightarrow \infty$.

Following the logic of \citet{fong_martingale_2024}, we express Monte Carlo uncertainty over the final choice of model in light of different possible imputations or realizations of the complete dataset. Rather than targeting a parameter $\theta_\infty$, our goal is $p(\mathcal{M}_k \mid \mathcal{D})$ directly on the model space. The benefit of this approach is that uncertainty arises from actual observables. The only required inputs are a routine to optimize a consistent model selection criterion, and a method to generate new samples from the model. We go from decision to uncertainty, making a fully determined choice for each replication of the data, rather than from uncertainty to decision, making a single choice based on the evidence for limited data.

\begin{center}
  \begin{minipage}{.9\linewidth}
    \singlespacing
    \begin{algorithm}[H]
      \caption{Predictive resampling}
      \begin{algorithmic}[1]
        \State Specify search space of candidate models $\{\mathcal{M}_k\}$
        \State Set number of trials $B$ and final sample size $N \gg n$
        \For{$b$ from 1 to $B$}
          \For{$i$ from $n+1$ to $N$}
            \State Calculate consistent model selection criterion $C(\mathcal{M}_{k(i-1)}, x_{1:i-1})$
            \State Optimize and identify best model $\bestmodel{i-1} = \argmax_k C(\cdot, \cdot)$ 
            \State (If applicable) Identify parameter MLE $\hat{\theta}_{k(i-1)}$
            \State Sample $x_i \sim p(\cdot | \bestmodel{i-1})$ and add to training data
          \EndFor
          \State Record final model $\bestmodel{N}$
        \EndFor
        \State Return posterior probabilities $p(\mathcal{M}_k \mid \mathcal{D}) = B^{-1} \sum_{b=1}^B \mathbb{1}(\bestmodel{\infty}^{(b)} = \mathcal{M}_k)$
      \end{algorithmic}
    \end{algorithm}
  \end{minipage}
\end{center}

For supervised data, such as observed covariates $X$ and outcomes $Y$, we view the design matrix as fixed, and the target for predictive resampling is the imputation of new outcomes. In other words, we copy $x_{1:n}$ as a block, yielding $x_{n+1:2n}$; predict $y_{n+1:2n}$ using the optimal fitted model; add $(x,y)_{n+1:2n}$ to the observed data; and so forth. This ensures that the resampling process only serves to propagate uncertainty, and does not introduce out-of-distribution bias in the covariate space. Sampling the entire outcome vector in blocks, rather than sampling individual observations, also keeps with the idea of ``repeating the experiment'' and simplifies the update calculation. We defer further discussion of this setup to the illustration of variable selection in Section \ref{sec:varsel}.

As a further note, given identifiable models, we highlight that this approach recovers a standard Bayesian update when using the usual posterior predictive to sample new observations. In other words, for Lines 5 to 8 in Algorithm 1, rather than optimizing and sampling from the best model, we would draw $x_i$ directly from the posterior predictive using the intermediate model mixing weights at step $i$ within that trial. Each of the $B$ trials would eventually converge to a single model; the relative proportions of these trials would converge to the initial model mixing weights as $B \rightarrow \infty$. The predictive framework can therefore be considered a generalization of standard Bayesian model uncertainty.

\subsubsection*{Demonstration}
\label{sec:simplehyp}

Suppose the models we want to compare are two point hypotheses $H_0: \theta = \theta_0$ (represented by $\mathcal{M}_0$) against $H_1: \theta = \theta_1$ (or $\mathcal{M}_1$) for the unknown mean parameter of a normal distribution with known variance $\sigma^2 = 1$. Since there is only a single parameter to be estimated, there is no penalty for model complexity -- we can directly maximize the log-likelihood
\begin{align*}
\bestmodel{n} &= \argmax_k C(\mathcal{M}_{k(n)}, x_{1:n}), \\
C(\mathcal{M}_{k(n)}, x_{1:n}) &= \log(\mathcal{L}) = \sum_{i=1}^{n} \log\,p(x_i \mid \theta_k)    
\end{align*}
over $k \in \{0, 1\}$ to select between hypotheses. Under the predictive resampling approach, we generate a new data point $x_{n+1} \sim \mathcal{N}(\theta_{\widehat{k}(n)}, 1)$, add it to the observed data, find the choice of model which maximizes the likelihood of the augmented data, and repeat the process. Once we have imputed a sufficiently large number of additional samples $N >> n$, we record our final choice of model, then repeat the pipeline, and index our uncertainty over replications of the ``completed'' population. 

We now illustrate the convergence properties of the model choice in this simplest setting. For $k \in \{0, 1\}$, consider the model likelihood
$$L_m(k)=\prod_{i=1}^{m}\frac{\mathcal{N}(x_i \mid \hat{\theta}_k,1)}{\mathcal{N}(x_i \mid \hat{\theta}_{(m-1)},1)}$$
where $\hat{\theta}_{(m-1)}$ maximizes $\prod_{i=1:m-1}\mathcal{N}(x_i\mid\theta,1)$ with $\theta \in \{\theta_0,\theta_1\}$.
Then
$$E\,(L_m(k) \mid x_{1:m-1})=\prod_{i=1}^{m-1}\frac{\mathcal{N}(x_i \mid \hat{\theta}_k,1)}{\mathcal{N}(x_i \mid \hat{\theta}_{(m-1)},1)}\leq 1,$$
and also
$E(L_m(k)\mid x_{1:m-1})\geq L_{m-1}(k)$. Hence, for both $k \in \{0,1\}$ it is that $L_m(k)$ converges due to the martingale convergence theorem. The parameter is therefore selected which maximizes $L_{\infty}(k)$; in the limit, $L_\infty(k)$ will either be 0 or 1. We provide a more general version of this proof in Section \ref{sec:mcproof}.

Experimentally, we generate $n=100$ data points from $\mathcal{N}(0, 1)$ as our observed data, which gives $\bar{x} = 0.066$. We set $\theta_0 = 0$ and vary $\theta_1$ in the range $\{$-0.3, -0.29, -0.28, $\ldots$, 0.3$\}$, and impute an additional $N = n+2500$ samples over $B = 1000$ trials. Figure \ref{fig:jellyLL} demonstrates that this is a sufficiently large sample size for the final model choice to converge, in the example case where $\theta_1 = 0.1$. For the initial iteration, we will always select the model where $\theta_k$ is closer to the observed sample mean of 0.066; however, from the second iteration forwards, as the imputed data points introduce uncertainty, it becomes possible to switch between models. We would expect to see that $H_1$ is chosen most frequently in cases where the alternative mean $\theta_1$ is closer to $\bar{x}$. Figure \ref{fig:demoLL} confirms this: with the baseline case of $\theta_0 = \theta_1 = 0$ marked at 50\% in light blue, $H_1$ is selected more frequently in cases where $\theta_1$ is closer to $\bar{x}$, and less frequently otherwise.

\begin{figure}
\centering
\includegraphics[width=\textwidth]{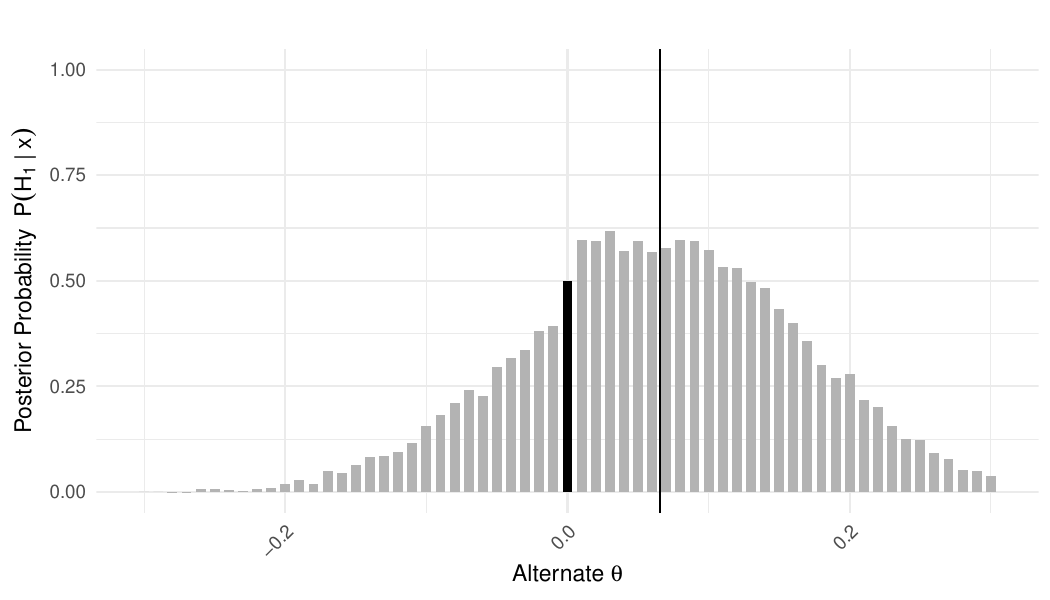}
\caption{Proportion of trials in which alternate model $H_1$ is selected as alternate mean $\theta_1$ varies from -0.3 to 0.3. The observed sample mean $\bar{x}$ is denoted by the vertical line, while the baseline of $\theta_0 = \theta_1 = 0$ is denoted by the black bar.}
\label{fig:demoLL}
\end{figure}

As an alternative means of explanation, Figure \ref{fig:jellyLL} captures the idea that model uncertainty can often be expressed in terms of a decision rule on some population-level summary statistic. In this case, when comparing $H_0: \theta_0 = 0$ against $H_1: \theta_1 = 0.1$, we would expect this decision threshold to be at $\bar{x} = 0.05$. We in fact see a clear divergence between populations with $\bar{x}_N > 0.05$, where we select $H_1$, and populations with $\bar{x}_N < 0.05$, where we select $H_0$. By targeting a probability distribution over the complete population $p(x_{n+1:\infty} \mid x_{1:n})$, what we are really attempting to retrieve is a probability distribution over this summary statistic computed on the infinite population. This partitions the space of possible observable datasets into critical regions, which we can then interpret in the context of model uncertainty.

\subsection{Consistent model selection}

This framework requires the specification of a model selection criterion, which will be used to select the best model at each step for the generation of $x_{n+1}$ given observed $x_{1:n}$, and ultimately to select the best final model for each possible realization of the complete data. The key requirement for this criterion is therefore that it be consistent, i.e. that the probability of selecting the correct model converges to 1 as $N \rightarrow \infty$ \citep{claeskens_model_2008}. The BIC provides consistency, as well as a clear connection to Bayesian inference, and therefore serves as our preferred one-step model selection rule \citep{schwarz_estimating_1978}.

To further motivate the BIC, we can consider its more general interpretation in the context of predictive evaluation. A scoring rule $S$ is a general summary measure for a probabilistic forecast \citep{matheson_scoring_1976}. Consider the score which is defined as the sum of the individual log-predictive probabilities
\begin{equation}
S(x_{1:n}, \mathcal{M}) = \sum_{i=1}^n \log p_{\mathcal{M}}(x_i \mid x_{1:i-1})
\end{equation}
with the data evaluated as if they had appeared in sequence. This is a proper scoring rule \citep{gneiting_strictly_2007} in that its expected value is maximized when the true distribution is used for forecasting. In fact, \citet{fong_marginal_2020} show that it is the unique scoring rule which guarantees coherent model evaluation. 

\citet{dawid_present_1984} shows that this scoring rule is equivalent to the Bayes factor, which the BIC directly approximates. This means that the difference in scores between models $\mathcal{M}_1$ and $\mathcal{M}_2$
\begin{equation}
\log BF = S(x_{1:n}, \mathcal{M}_1) - S(x_{1:n}, \mathcal{M}_2)
\end{equation}
can be interpreted as the ``weight of evidence'' in favor of the first model \citep{good_rational_1952}. The BIC therefore provides a general method for comparing forecasting rules \citep{gneiting_strictly_2007}, supporting its use as a tool to select the best model for a one-step predictive update $p(x_{n+1} \mid x_{1:n})$. 

We now consider other common approaches for model comparison. Cross-validation splits the observed data (of size $n$) into two parts, uses the first $n-p$ samples to train the model, and then evaluates predictive performance on the held-out validation set of $p$ samples. Full leave-$p$-out cross-validation requires fitting the model $\binom{n}{p}$ times, which can be computationally infeasible, and so simplifications such as $k$-fold cross-validation (dividing the data into blocks of size $n/k$ for testing) are often preferred \citep{geisser_predictive_1975}. The marginal likelihood is equivalent to exhaustive leave-$p$-out cross-validation averaged over all values of $p$ and all possible test sets, providing an underlying link between these approaches \citep{fong_marginal_2020, gneiting_strictly_2007}.

One of the most common simplifications is leave-one-out cross-validation (LOO-CV), where the model is fit and evaluated $n$ times. This is equivalent to $k$-fold cross-validation with $k=n$. However, \citet{shao_linear_1993} shows that LOO-CV is asymptotically inconsistent for linear models, meaning it should not be applied in a context where our eventual target is a decision based on the complete data. This can be rectified by using leave-$p$-out cross-validation with $p/n \rightarrow 1$ as $n \rightarrow \infty$, where the size of the validation set grows alongside the sample size, though this again faces computational constraints as $n$ grows. Intuitively, this is necessary because a larger validation set provides a smoother assessment of prediction error; optimizing the fit on a single observation at a time can select unnecessarily large models. \citet{yang_consistency_2007} extends these findings to nonparametric models, while \citet{vehtari_bayesian_2002} motivate cross-validation in the Bayesian context. \citet{vehtari_practical_2017} and \citet{sivula_uncertainty_2023} discuss Bayesian LOO-CV; it is noted to be unreliable in certain common use cases, such as comparing similar models or misspecified models. 

The AIC approximates predictive fit, given by
\begin{equation}
\text{AIC} = 2d-2\log \widehat{\mathcal{L}}
\end{equation}
where the penalty for model complexity is fixed, unlike the variable penalty of $\log n$ applied in the BIC \citep{akaike_new_1974}. However, the AIC is asymptotically equivalent to LOO-CV, and so it is also asymptotically inconsistent \citep{shao_linear_1993}, for the similar reason that it tends to overfit to models that are too complex. 

The Lasso encourages sparsity in a linear model by imposing an $\mathcal{L}_1$ penalty on the regression coefficients $\beta$, which performs variable selection by driving a subset of coefficients to zero \citep{tibshirani_regression_1996}. For regression data $\{X, y\}$, this is done with the objective function
\begin{equation}
\min_{\beta} \left( \sum_{i=1}^{n} (y_i - X_i \beta)^2 + \lambda \sum_{j=1}^{p} |\beta_j| \right)
\end{equation}
where $\lambda$ controls the strength of the penalty. However, \citet{zhao_model_2006} show that the model selection provided by Lasso is only consistent under certain conditions relating to the covariance of $X$ (specifically, that the variables outside the true model cannot be represented by those in the true model). 

Overall, our approach is grounded in the idea that statistical uncertainty stems from missing data, and if complete data were available, we could reliably identify the correct model. For well-calibrated uncertainty propagation, we need a criterion that can consistently choose the right model once the missing observations are recovered. While this does not necessarily have to be the BIC, it should not be methods like LOO-CV, AIC, or Lasso, as they lack the guarantee of asymptotic consistency. In Section \ref{sec:varsel}, we demonstrate in the context of variable selection that AIC and Lasso select overly complex models as expected when used for predictive resampling in the context of variable selection.

\section{Hypothesis testing as a decision problem}
\label{sec:hyptest}

We now return to hypothesis testing, perhaps the most prevalent model uncertainty question in the statistical literature. We begin with a brief discussion of certain issues in both frequentist and Bayesian hypothesis testing, along with recent work on $e$-values. In particular, we focus on how the interpretation of these results can often be unclear or counter-intuitive. We then demonstrate how predictive resampling frames hypothesis testing as a decision problem on a population statistic, propagating uncertainty through the missing data to directly quantify the probability of competing hypotheses, without requiring the specification of a prior.

\subsection{Frequentist methods and $p$-values}

Suppose we are testing $H_0: \theta \in \Theta_0$ against $H_1: \theta \in \Theta_1$ for some parameter $\theta$ whose true value is unknown. Generally, we would define our test such that $\Theta_0 \cap \Theta_1 = \emptyset$ and $\Theta_0 \cup \Theta_1 = \Theta$. A classical (frequentist) procedure constructs a critical region $C$ such that the null hypothesis $H_0$ is rejected if the observed sample $Y_{1:n}$ falls within $C$ and not rejected if $Y \notin C$ \citep{berger_could_2003}. Typically, this is done by calculating a test statistic $T = t(Y)$ as a function of the observed sample, then specifying the critical region as $C = \{y \mid t(y) > t_0\}$. The critical value $t_0$ is set using a pre-specified significance level of $1-\alpha$ such that $\alpha = P_0(T > t_0)$ is the selected Type I error rate, or the probability of erroneously rejecting the null hypothesis when it is actually true. 

A key aspect of this approach is that it focuses solely on the observed $Y_{1:n}$ to calculate the test statistic, which is then treated as a random variable. Within this construction, the $p$-value $p = P_0(T > t(y))$ can be understood as the probability under the null hypothesis of obtaining a test statistic which is ``at least as extreme'' as the observed value. This probability is interpreted with respect to the variation across all possible samples of size $n$ in the population from which $Y$ is drawn. For a specific observed sample, having $t(y) > t_0$ is then equivalent to having $p < \alpha$, and so hypothesis tests are often discussed in terms of whether the $p$-value falls below the chosen significance threshold.

However, a well-known issue with rejecting the null hypothesis based on $p < \alpha$ at sample size $n$ is that this conflates the effect size with the sample size \citep{nickerson_null_2000, gelman_difference_2006}. For example, suppose we are testing $H_0: \theta=0$ against $H_1: \theta \neq 0$ for the unknown mean parameter of a normal distribution with known variance $\sigma^2 = 1$. For an observed sample $Y_{1:n}$, the effect size would be calculated as the mean difference from the null hypothesis value of 0, or $\bar{y}-0 = \bar{y}$. As the sample size increases to $\infty$, the distribution of possible values of the test statistic narrows with respect to the effect size, ultimately converging to a point mass at 0. Accordingly, since the tail area of this distribution contains fewer ``extreme'' outcomes, a vanishingly small effect size becomes sufficient to establish a ``significant'' difference between $H_0$ and $H_1$, even when that difference may be below a threshold of practical relevance.  Figure \ref{fig:pvn} demonstrates how, as the sample size grows, a $p$-value below the standard 0.05 threshold is returned for smaller and smaller effect sizes. As $n$ approaches infinity, any point null hypothesis will be rejected with probability 1. \citet{jeffreys_theory_1961} famously summarizes this result, emphasizing why $p$-value computations based on tail area are illogical: ``...a hypothesis that may be true may be rejected because it has not predicted observable results that have not occurred.''

\begin{figure}
\centering
\includegraphics[width=\textwidth]{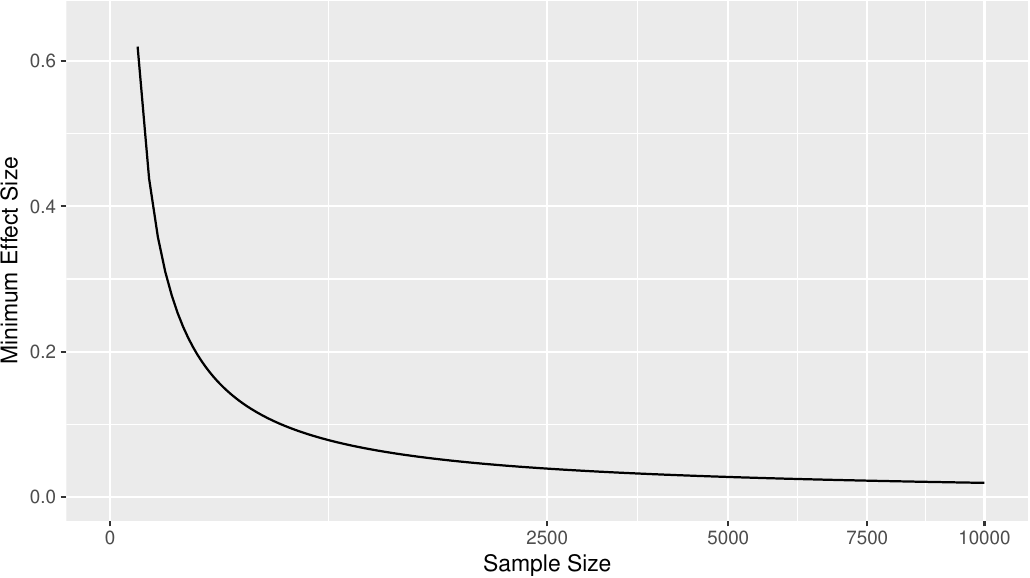}
\caption{Minimum effect size $\bar{y}-\theta$ sufficient with sample size $n$ to reach ``significant'' conclusion of $p<0.05$ for simple one-sample z-test.}
\label{fig:pvn}
\end{figure}

\subsection{Bayesian methods and $e$-values}

Meanwhile, Bayesian hypothesis testing can face challenges related to the precise specification of prior densities $\pi_0$ and $\pi_1$ for the null and alternative hypotheses respectively. A famous illustration is the Jeffreys-Lindley paradox \citep{lindley_statistical_1957, jeffreys_theory_1961}. Briefly, consider testing the null model $H_0: \theta = 0.5$ against the alternative model $H_1: \theta \neq 0.5$ for a binomial proportion, such as the fraction of coin flips that come up heads. Under a uniform prior of $\pi_0 = \pi_1 = 0.5$, $H_0$ will be preferred even when the observed data would support a frequentist $p$-value well below 0.05; intuitively, this is because $H_1$ is much more diffuse, with prior probability spread across $\theta \in \left[0, 1\right]$, leading to a relatively lower posterior probability regardless of the observed data. However, \citet{johnson_use_2010} notes that the opposite issue arises when tests are defined such that $\pi_1$ is positive on $\Theta_0$, which they refer to as a ``local alternative prior density". The lack of a significant separation between the two hypotheses results in evidence being accumulated much more rapidly to support true alternative models than to support true null models, making it asymmetrically difficult to reject $H_0$.

In recent years, an approach for Bayesian-adjacent hypothesis testing has emerged under the name of ``$e$-values'', primarily motivated by this problem of combining results from multiple studies \citep{vovk_e-values_2021, grunwald_safe_2024}. The $e$-value is closely linked to the likelihood ratio, and certain optimal $e$-values are identical to Bayes factors with specifically chosen priors \citep{grunwald_safe_2024}. \citet{shafer_testing_2021} discusses a very similar concept using the terminology of ``betting scores'', arguing that the $e$-value is most naturally interpreted as the monetary winnings from making a bet against the null hypothesis.

The $e$-value provides several concrete benefits with respect to accumulating evidence against a null hypothesis, where the results from sequential experiments can be combined while still maintaining strict control of Type I error rate. However, the direct interpretation of an $e$-value remains unclear. Generally, for a test based on an infinite sample, we would want a test statistic $T_n \in (-\infty, \infty)$ to behave as $T_n \to -\infty$ under the null and $T_n \to \infty$ under the alternative. If the convergences are at the same rate, then the critical value can be set at 0. For example, if we are testing a normal mean with $H_0:\theta=0$, then one such test statistic is
$$T_n=\log (\bar{y}_n^2\,\sqrt{n})=\left\{
\begin{array}{ll}
\half\log n+ \log(\bar{y}_n^2) & H_1:\theta\ne 0 \\
-\half\log n+\log (n\bar{y}_n^2) & H_0:\theta=0.
\end{array}
\right.$$
Under $H_1$, $T_n$ goes to infinity at rate $\log n$ since $\bar{y}_n^2$ is converging to a non-zero finite constant, whereas under $H_0$, $T_n$ goes to minus infinity at rate $-\log n$ since $n\bar{y}_n^2$ is standard normal and therefore finite. Hence, a suitable critical value is 0. This amounts to a modified version of the BIC. 

In contrast, the $e$-value starts with a statistic $T_n$ and the assumption that under the null, $E(T_n) \leq 1$. The next step often still relies on heuristic evaluation; \citet{vovk_e-values_2021} resort to the rule of thumb provided by \citet{jeffreys_theory_1961} for Bayes factors (values between 1 and $\sqrt{10}$ are ``not worth more than a bare mention''; between $\sqrt{10}$ and 10 are ``substantial''; etc.)

In contrast, predictive resampling can provide a direct probability both for and against a null hypothesis, without requiring the specification of prior probabilities. Having observed data $Y_{1:n}$, we do not take the usual frequentist approach of viewing it as one possible realization of a random variable, but also do not assign a prior probability over the hypothesis space. Instead, as before, we focus on the missing $Y_{n+1:\infty}$ as the source of uncertainty, where the decision rule for selecting a hypothesis is based on some summary statistic of the complete population. We compare the null and alternative hypotheses at each step (by a consistent model selection criterion such as BIC) and sample a new observation from the preferred hypothesis, including any necessary parameter MLE(s). We then index the posterior probability of each hypothesis over repeated replications of the complete data.

\subsubsection*{Demonstration}

To demonstrate, consider the previous point null hypothesis test of $H_0: \theta = \theta_0$ against $H_1: \theta \neq \theta_0$ for the unknown mean parameter of a normal distribution with known variance $\sigma^2 = 1$. The $p$-value is calculated using a classical two-sided z-test. Following \citet{vovk_e-values_2021}, we calculate the $e$-value as the likelihood ratio
$$E = \frac{e^{-(y-\delta)^2/2}}{e^{-y^2/2}} = e^{y\delta - \delta^2/2}$$
which can be transformed to the $[0, 1]$ window if desired using the ``e-to-p calibrator'' function $p = \min(1, 1/e)$. 

For the predictive resampling approach, the null model is fully determined, while the alternative model has one free parameter ($\theta = \bar{y}$) for a BIC penalty term of $d=1$. For observed data $y_1, \ldots, y_n$, comparing the BIC therefore reduces to rejecting $H_0$ if $n\bar{y}^2_n>\log n$, and vice versa. At each step, we sample $y_m$ from $\mathcal{N}(0,1)$ if $(m-1)\bar{y}^2_{m-1}<\log(m-1)$ and from $\mathcal{N}(\bar{y}_{m-1},1)$ otherwise. The resulting sequence 
$$h_{m+1}=1(m\bar{y}^2_{m}>\log m),\quad m>n$$
converges to 0 or 1. We repeat this up to the large sample size of $N = n + 20n$ across several Monte Carlo iterations (in this case, 1000) to index our uncertainty between the two hypotheses. In the supplementary material, we provide convergence diagrams showing that this value of $N$ is sufficiently large for the choice of model to converge.

\begin{table}
\caption{\label{tab:hyptest}Metrics for two-sided hypothesis testing demonstration across 100 random seeds}
\centering
\begin{tabular}{cccccc}
\toprule
\multirow{2}{*}{True model} & \multirow{2}{*}{Summary metric} & \multicolumn{4}{c}{Sample size} \\ 
    & & 30 & 100 & 300 & 1,000 \\ 
\midrule
\multirow{5}{*}{$H_0$ ($\mu = 0$)} & Prop. of tests with $p<0.05$ (Type I error) & 6\% & 2\% & 7\% & 4\% \\
    & Prop. of tests with $e>10$ (Type I error) & 3\% & 1\% & 5\% & 3\% \\
    \cmidrule{2-6}
    & Average resampling posterior prob. of $H_1$ & 0.15 & 0.07 & 0.05 & 0.03 \\ 
    & Prop. of tests with $P(H_1 \mid \mathcal{D}) > 0.05$ & 83\% & 34\% & 19\% & 8\% \\
    & Prop. of tests with $P(H_1 \mid \mathcal{D}) > 0.1$ & 37\% & 14\% & 11\% & 4\% \\
\midrule
\multirow{5}{*}{$H_1$ ($\mu = 0.1$)} & Prop. of tests with $p<0.05$ (Power) & 7\% & 18\% & 37\% & 85\% \\
    & Prop. of tests with $e>10$ (Power) & 5\% & 13\% & 29\% & 81\%  \\
    \cmidrule{2-6}
    & Average resampling posterior prob. of $H_1$ & 0.17 & 0.18 & 0.27 & 0.64 \\
    & Prop. of tests with $P(H_1 \mid \mathcal{D}) > 0.5$ & 10\% & 13\% & 22\% & 65\% \\
    & Prop. of tests with $P(H_1 \mid \mathcal{D}) > 0.9$ & 4\% & 4\% & 10\% & 44\% \\
\bottomrule
\end{tabular}
\end{table}

We first generate data under the null, i.e. from $\mathcal{N}(0, 1)$, for sample sizes $n = \{30, 100, 300, 1000\}$ across 100 random seeds. Summary metrics for the three methods described above can be found in Table \ref{tab:hyptest}. As expected, the classical $p$-values under the null are uniformly distributed across $[0,1]$ regardless of the observed sample size. Since the $p$-values cannot be directly interpreted as probabilities, we interpret them as a binary decision using a pre-specified Type I error rate, usually $\alpha = 0.05$. In the first row of Table \ref{tab:hyptest}, we see that the proportion of tests with $p<0.05$ is approximately 5\% for all values of $n$. For the $e$-values, as recommended by \citet{vovk_e-values_2021}, we apply Jeffreys' rule of thumb, under which $e > 10$ indicates that the evidence against the null hypothesis is ``strong". The proportion of tests with $e>10$ also remains at approximately the same level for all values of $n$.

In contrast, predictive resampling allows us to accumulate evidence \emph{in favor of the null} as the observed sample size grows. In the third row of Table \ref{tab:hyptest}, we see that the average posterior probability of $H_1$ returned by resampling decreases with $n$. This reflects the underlying principle that missing data is the source of statistical uncertainty, and that observing additional data consistent with a given model (in this case, $H_0$) should result in a greater degree of certainty about that model. Rather than requiring a heuristic rule to provide a binary decision, the predictive framework enables a richer expression of uncertainty over different possible completions of the data, which can be directly interpreted as a posterior probability.

If desired, we can also apply a decision rule, and record the proportion of trials for which the posterior probability of $H_1$ exceeds a certain level. In the fourth and fifth rows of Table \ref{tab:hyptest}, we see that under the null hypothesis, the proportion of tests with $P(H_1 \mid \mathcal{D})$ above certain thresholds (in this case, 0.05 and 0.1) decreases with $n$. Conversely, this means that increasing the sample size results in a greater proportion of tests with at least 0.95 and 0.9 posterior probability on $H_0$ respectively.

As a different mode of visualization, in Figure \ref{fig:h0demo} we plot the classical $p$-values for each seed on the horizontal axis, and the resampling posterior probabilities of $H_0$ on the vertical axis. The X marks on the plots indicate tests with $p<0.05$ where classical testing would reject $H_0$, and the O marks indicate $p>0.05$ for which classical testing would fail to reject $H_0$. We again see that the $p$-values are invariant to sample size, but that the overall level of the resampling probabilities increases towards 1 as the sample size grows. In the supplementary material, we provide a corresponding plot comparing $e$-values with predictive resampling; the interpretation is largely similar.

Alternatively, in the bottom half of Table \ref{tab:hyptest}, we generate the same sample sizes from $\mathcal{N}(0.1, 1)$ under $H_1$. All three methods successfully accumulate evidence against the null as $n$ increases; this can be seen visually in Figure \ref{fig:h1demo}, where the points gradually migrate towards the bottom and left as we observe more data. (The corresponding plot for $e$-values is in the supplementary material.)

\begin{figure}
\centering
\includegraphics[width=\textwidth]{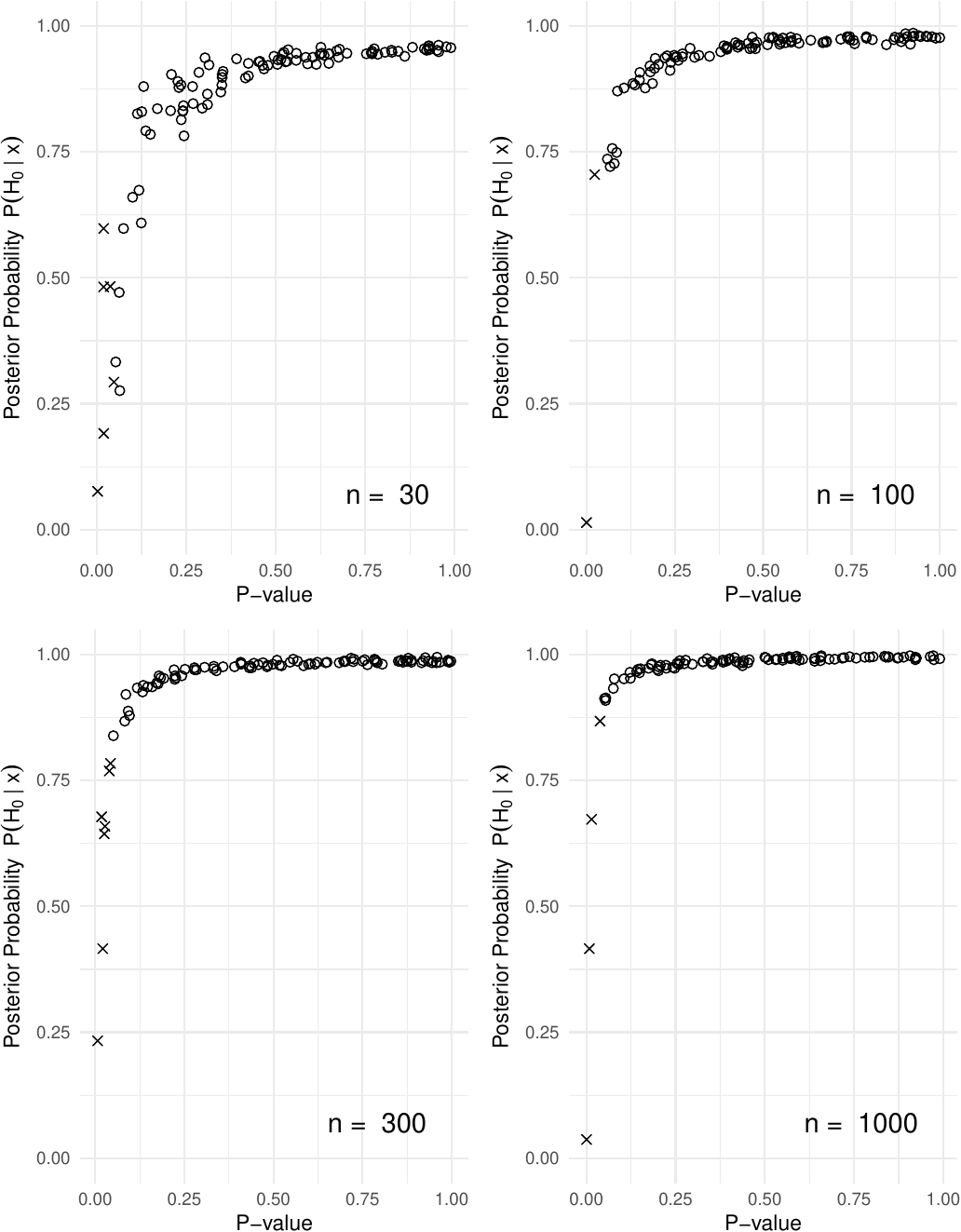}
\caption{Observed $p$-value vs. resampling posterior probability of $H_0$ for data generated under the null $\mathcal{N}(0, 1)$ across 100 random seeds. X denotes tests with $p < 0.05$ where classical testing rejects $H_0$, and O denotes tests with $p > 0.05$.}
\label{fig:h0demo}
\end{figure}

\begin{figure}
\centering
\includegraphics[width=\textwidth]{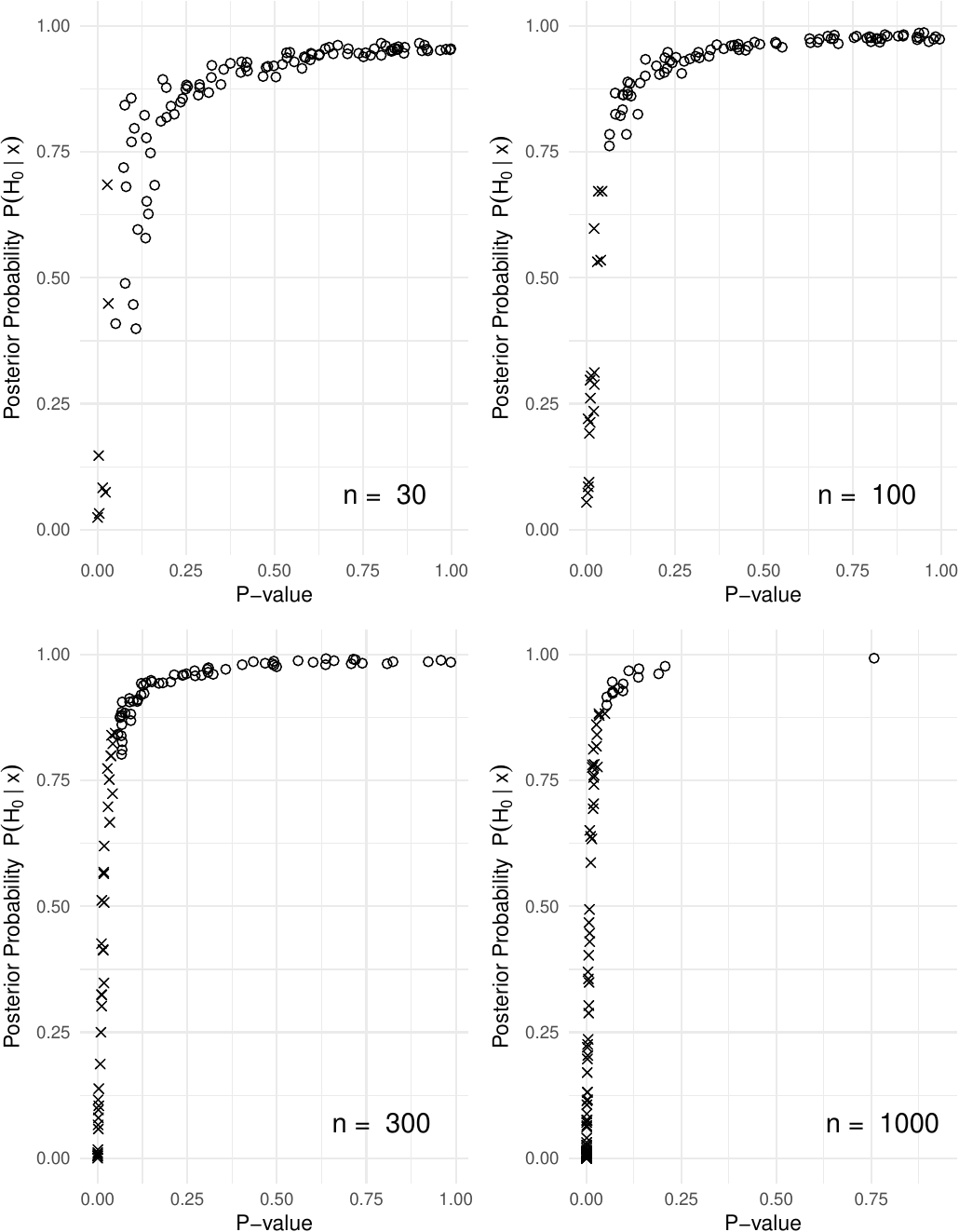}
\caption{Observed $p$-value vs. resampling posterior probability of $H_0$ for data generated under the alternative $\mathcal{N}(0.1, 1)$ across 100 random seeds. X denotes tests with $p < 0.05$ where classical testing rejects $H_0$, and O denotes tests with $p > 0.05$.}
\label{fig:h1demo}
\end{figure}

In this context, since $H_1$ is the ``true'' underlying model, we would interpret the $p$-value and $e$-value results in terms of statistical power. As expected, both increase in power as $n$ grows, based on the proportion of tests with $p<0.05$ or $e>10$ respectively. However, this calculation is only possible because we are simulating several datasets. The concept of power does not directly translate to a single trial, where we only get one chance to observe the outcome; it is inherently probabilistic, estimating the likelihood that a study will detect a certain effect across many repetitions of the same trial. This can be observed in applied statistical practice, where power calculations generally have to resort to simulation.

In contrast, the average posterior probability of $H_1$ returned by predictive resampling also grows as $n$ increases, but the results for a single trial can be directly interpreted. Rather than having to interpret our outcomes in the context of long-run Type I and Type II error rates over repeated trials, we can index the probability of $H_1$ for each trial independently, providing a clearer picture of uncertainty. However, as before, we can also apply a decision rule if desired, and consider the proportion of individual trials for which $P(H_1 \mid \mathcal{D})$ exceeds a given level. To ``accept" the alternative hypothesis, we might specify a higher threshold (such as 0.5 or 0.9); at the bottom of Table \ref{tab:hyptest}, we see that these proportions increase with $n$.

\section{Convergence of one-step updates}

In this section, we demonstrate the convergence of the model choices returned by predictive resampling. We begin with a brief discussion of the martingale posterior distribution framework from \citet{fong_martingale_2024}, then extend it to one-step model updates.

\subsection{Exchangeability and martingales}

Suppose we observe $Y_{1:n}$ from an unknown true sampling distribution and wish to conduct inference on some parameter $\theta$. The typical Bayesian approach is to elicit a prior density $\pi(\theta)$ and likelihood function $f(y \mid \theta)$, then use this to derive the posterior density $\pi(\theta \mid y_{1:n})$. However, the statistical uncertainty in our posterior arises entirely from the fact that we only observe a sample of size $n$ and are therefore missing observations. If we knew the values $Y_{n+1:\infty}$, then any identifiable $\theta$ would be fully defined as some function of the complete data. An alternative to the standard approach is therefore to directly model the joint density $p(y_{n+1:\infty} \mid y_{1:n})$ and impute the missing observations, allowing us to recover the parameter of interest from the completed dataset. 

Under this construction, once we have imputed $Y_{n+1:N}$ at any intermediate point, we denote the estimated posterior mean as $\bar{\theta}_N = E(\theta \mid Y_{1:N})$. \citet{doob_application_1949} shows that if $\bar{\theta}_N$ is a martingale with
\begin{equation}
    E(\bar{\theta}_N \mid Y_{1:N-1}) = \bar{\theta}_{N-1}
\end{equation}
then, under certain identifiability and measurability conditions, $\bar{\theta}_N \rightarrow \theta$ almost surely as $N \rightarrow \infty$. In other words, the imputation of data $Y_{n+1:\infty}$ using the posterior predictive returns the same uncertainty as the typical prior-posterior calculation on $\theta$, and our imputed $\bar{\theta}_\infty$ is a draw from the posterior $\pi(\theta \mid y_{1:n})$.

To put this approach in context, recall that by the representation theorem of \citet{de_finetti_prevision_1937}, the typical Bayesian likelihood-prior setup comes from the assumption of exchangeability, where the joint probability of $Y_{1:m}$ for all $m$ does not depend on the ordering of the observations. The key insight of the result above is that we can weaken or relax exchangeability and still conduct full Bayesian inference. Specifically, Doob's martingale condition requires only that the imputed observations $y_{n+1:\infty}$ are conditionally identically distributed (c.i.d.), as described by \cite{berti_limit_2004}, with all future data points being identically distributed conditional on the observed $Y_{1:n}$. 

Without exchangeability, we lose the likelihood and prior, and so we instead turn to a predictive representation of Bayesian inference. Using the one-step update
\begin{equation}
    p(y_{n+1} \mid y_1, \ldots, y_n) = \frac{p(y_1, \ldots, y_n, y_{n+1})}{p(y_1, \ldots, y_n)}
\end{equation}
we view Bayes purely through the sequential learning of predictive distribution functions \citep{dawid_present_1984}. Having observed the first data point $y_1$, we update our beliefs in order to predict $p(y_2 \mid y_1$); after observing $y_2$, we again update our beliefs to predict $p(y_3 \mid y_1, y_2)$; and so on. We continue the process to model 
\begin{equation}
    p(y_{n+1:\infty} \mid y_{1:n}) = \prod_{i=n+1}^\infty f(y_i \mid y_1, \ldots, y_{i-1})
\end{equation}
and generate $y_{1:\infty} = (y_{1:n}, y_{n+1:\infty})$, from which we can then calculate $\theta$ as some functional of the data.

\citet{fong_martingale_2024} bring these insights together to construct a general predictive resampling process for the posterior uncertainty in $\theta$ with i.i.d. data. Under the previous result from Doob, martingales are applied to ensure convergence of the parameter of interest. The pipeline consists of two steps: simulating $Y_{n+1:\infty}$ with one-step predictive updates, as seen in Equation \ref{eq:update1}, and solving for the parameter $\theta$ from the simulated complete information via the limiting point estimate $\bar{\theta}_\infty$. To recover a general parameter of interest, the random limiting empirical distribution function $F_\infty$ can be used to compute $\theta_\infty = \theta(F_\infty)$, which is designated as the martingale posterior distribution.

Besides de Finetti, several prior works have similarly emphasized uncertainty quantification through observables rather than parameters. \citet{geisser_predictive_1975, geisser_aspects_1982} argues that exchangability implies the notion of parameters as relevant constructs, but that prediction is more appropriate than parametric inference since results are expressed in terms of actual observables. \citet{dawid_well-calibrated_1982, dawid_present_1984} goes one step further to argue that prediction should be the primary focus of statistical inference, and that the goal at any point is to manipulate the currently observed data in order to form a well-calibrated probability distribution for the next observation.

\citet{roberts_probabilistic_1965} presents a version of the predictive argument for finite data, noting that any statement about finite population parameters can be reinterpreted as a predictive statement about the unobserved part of the population. The simplest model for this broad approach (repeatedly sampling observables, then estimating parameters with the complete data) will be familiar to readers, where the induced distribution function on the object of interest is given by 
\begin{equation}
F(y)=n^{-1}\sum_{i=1}^n w_i\,1(y_i \leq y)
\end{equation}
with random weights $w_i \sim \text{Dirichlet}(1, \ldots, 1)$, and so the explicit sampling of $y_{n+1:\infty}$ is not actually required. This is the Bayesian bootstrap of \citet{rubin_bayesian_1981}, and the corresponding $p(y_{n+1:\infty} \mid y_{1:n})$ is known as the P\'olya urn sampling scheme.

In Bayesian nonparametrics, \citet{fortini_predictive_2012, fortini_quasi-bayes_2020} analyze the predictive construction of underlying models, where resampling is applied to retrieve the prior law of the mixing distribution. \citet{berti_class_2021} discuss a general class of models building on c.i.d. sequences, while \citet{hahn_recursive_2018} provide a fast online approach for sequential updates that makes use of bivariate copulas. Instructive reviews of the predictive approach to Bayesian inference are given by \citet{berti_probabilistic_2023} and \citet{fortini_exchangeability_2024}; as above, these rely on the overall view of Bayes as a sequential learning problem, where the goal is to specify one-step marginal predictive updates without having to rely on the usual prior-posterior pipeline.

\subsection{Extension to model uncertainty}
\label{sec:mcproof}

In the context of model selection, we now have candidate models $\{\mathcal{M}_k\}_{k=1}^K$ for finite $K$. Given observed data $x_{1:m}$, let $\bestmodel{m}$ be the optimal candidate model under a consistent model selection criterion, with any necessary parameter MLE(s) denoted by $\bestparam{m}$. We sample $x_{m+1}$ from the predictive $p(\cdot \mid \bestmodel{m}, \bestparam{m})$, append it to our dataset yielding $x_{1:m+1}$, and repeat this process. 

As before, the key requirement is that the model choice $\widehat{k}(m)$ converges to some $k(\infty)$ as the sample size grows from $m \rightarrow \infty$. We demonstrate this first for the case where the intermediate models are selected by maximum likelihood, and then for the case where this likelihood is penalized (as in the AIC, BIC, or LASSO). This is a general version of the specific proof provided in the demonstration from Section \ref{sec:mcmethod}.

\begin{prop} 
The model $\bestmodel{m}$ selected by sequential maximum likelihood converges as $m \to \infty$.
\end{prop}

\begin{proof} 
The (maximum) likelihood for model $k$ at stage $m$ is given by
$$\prod_{i=1}^m p(x_{i} \mid \model{m}, \param{m})$$
which we normalize with respect to the previous best model
$$L_m(k)=\prod_{i=1}^m \frac{p(x_i \mid \model{m}, \param{m})}{p(x_i \mid \bestmodel{m-1}, \bestparam{m-1})}$$
We know that
$$\prod_{i=1}^{m-1} p(x_i \mid \bestmodel{m-1}, \bestparam{m-1}) \geq \prod_{i=1}^{m-1} p(x_i \mid \bestmodel{m-2}, \bestparam{m-2})$$
by the definition of $\bestmodel{m-1}$, and so
$E[L_m(k) \mid x_{1:m-1}] \leq L_{m-1}(k)$. 
By the convergence of non-negative supermartingales, $L_m(k)$ then converges for each $k \in K$. Since $|K|<\infty$ it follows that if we take $k(\infty)$ to maximize $L_\infty(k)$, then $\widehat{k}(m) \to k(\infty)$ almost surely.
\end{proof}

Generally, a model is selected not only by maximizing the likelihood function, but with some additional notion of penalty, as seen in the AIC, BIC, and Lasso. In each case, the model and parameter are selected which maximize
$$c(m, d_k, \theta_{k(m)}) \, \prod_{i=1}^m p(x_i \mid \model{m}, \theta_{k(m)})$$
where $c$ is a penalty function in terms of the sample size $m$, model dimension $d$, and parameter $\theta$. For AIC, $c(m,d,\theta)=e^{-d}$; for BIC, $c(m,d,\theta)=e^{-d\log m}$; and for Lasso, $c(m,d,\theta)=e^{-\lambda_m|\theta|}$, for some increasing $\lambda_m>0$. 

\begin{prop} 
The model $\bestmodel{m}$ selected by sequential penalized maximum likelihood converges as $m \to \infty$.
\end{prop}

\begin{proof} 
We now consider
$$L_m(k, \theta_{k(m)}) = c(m, d_k, \theta_{k(m)}) \, \prod_{i=1}^m \frac{p(x_i \mid \model{m}, \theta_{k(m)})}{p(x_i \mid \bestmodel{m-1}, \bestparam{m-1})},$$
and
$$E(L_m(k, \theta_{k(m)}) \mid x_{1:m-1}) = L_{m-1}(k, \theta_{k(m)}) \, \frac{c(m, d_k, \theta_{k(m)})}{c(m-1, d_k, \theta_{k(m-1)})}.$$
This remains a supermartingale when $c$ decreases as $m$ increases, which is the case for the key model selection criteria listed. Hence, for each $(k, \theta_{k(m)})$, we have $L_m(k, \theta_{k(m)}) \to L_\infty(k, \theta_{k(\infty)})$ almost surely for some $L_\infty$. To extend this to uniform convergence, similar model conditions as for the convergence of an MLE are required, namely that each $\Theta_k$ is a compact space and each $p(x \mid \mathcal{M}_{k(m)}, \theta_{k(m)})$ is suitably regular.
\end{proof}

\section{Illustrations}

In this section, we demonstrate model uncertainty via predictive resampling on example problems from density estimation, variable selection, and multi-level hypothesis testing. Code for all illustrations can be found at https://github.com/vshirvaikar/MPmodel.

\subsection{Density estimation}
\label{sec:density}

\subsubsection*{Demonstration}

A typical model selection question is the number of components required in a finite Gaussian mixture model (GMM) for density estimation. To begin, we demonstrate the predictive resampling framework on a simple univariate example. We simulate $n=50$ samples from a GMM with $G=2$ components
$$f_0(y) = \half \mathcal{N}(y \mid -1, \sigma^2) + \half \mathcal{N} (y \mid 1, \sigma^2)$$
and rescale the individual data points to vary the standard deviation between $\sigma = 0.5$ and $\sigma=1$. Figure \ref{fig:density1} shows sample kernel density plots with a fixed bandwidth of 0.5 for the generated data. The two separate peaks are clearly visible when $\sigma=0.6$ (Figure \ref{fig:density1A}), for example, but begin to merge together when $\sigma=0.9$ (Figure \ref{fig:density1B}).

\begin{figure}
\centering
   \subfloat[Density with $\sigma=0.6$ for components\label{fig:density1A}]{
     \includegraphics[scale=0.85]{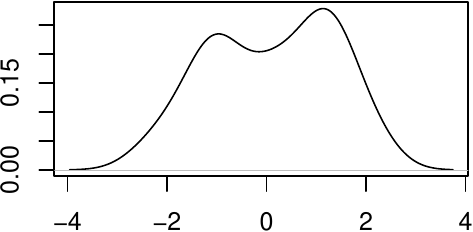}
   }
   \hfill
   \subfloat[Density with $\sigma=0.9$ for components\label{fig:density1B}]{
     \includegraphics[scale=0.85]{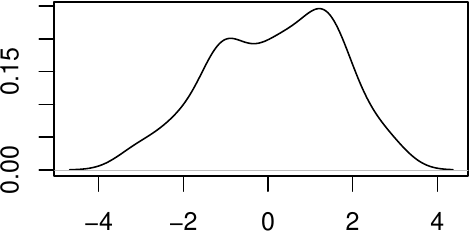}
   }
\caption{Kernel density plots for data generated from GMM with 2 components} 
\label{fig:density1}

\vspace{5cm}

\centering
\includegraphics[width=\textwidth]{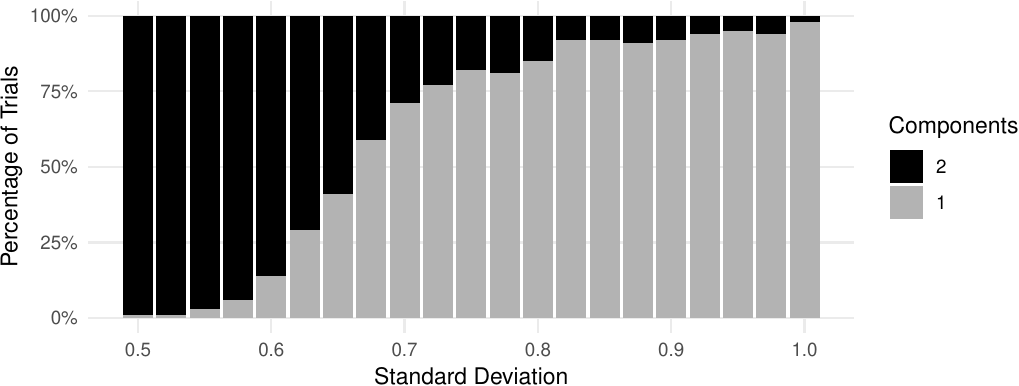}
\caption{Posterior uncertainty over number of components $G$ via resampling}
\label{fig:densmp1}
\end{figure}

Under the predictive resampling framework, the goal is to identify an optimal GMM at each step, which is then used to recursively predict one additional data point. Following \citet{fraley_model-based_2002}, we apply the expectation-maximization (EM) algorithm for clustering to our observed data of size $n$. The convergence properties of EM for Gaussian mixtures have been widely studied \citep{xu_convergence_1996}, and while its consistency is not guaranteed in all conditions, it is well-established for simple and correctly specified models such as the ones we explore here \citep{balakrishnan_statistical_2017}. 

We vary the number of components $G$ across a specified range and select the model with the lowest BIC. We assume equal variances, with a penalty of 
$$d = G \text{ means} + (G-1) \text{ proportions} + 1 \text{ common variance}$$
in the BIC calculation for a total of $2G$; if we allow unequal variances, the final term becomes $G$ as well for a total of $3G-1$. We simulate a new data point from this model, augmenting our dataset to size $n+1$, and repeat the above process. This continues for several additional points, yielding a final resampled dataset of size $N$, after which we record the final selected model for the ``complete'' data. We replicate this across several trials and then index our uncertainty over the distribution of final models.

We implement this pipeline using the \textbf{mclust} package in R \citep{scrucca_model-based_2023}, with candidate models containing either 1 or 2 components, and recursively simulate $N=n+200$ new observations per trial across a total of $B=100$ trials for each value of $\sigma$. We empirically find that this value of $N$ is sufficiently large for the model $\bestmodel{N}$ to closely approximate the final $\mathcal{M}_{k(\infty)}$; convergence diagrams can be found in the supplementary material. Figure \ref{fig:densmp1} shows the distribution of the final number of components across all trials, as $\sigma$ increases from 0.5 to 1. As expected, the posterior probability of selecting only $G=1$ component increases as the observed data becomes more unimodal.

\subsubsection*{Simulated example}

As a more complex example, we generate $n=20$ and $n=50$ data points from a GMM with $G=3$ components
$$f_0(y) = 0.4 \mathcal{N}(y \mid -3, 1) + 0.3 \mathcal{N} (y \mid 0, 1) + 0.3 \mathcal{N} (y \mid 4, 1)$$
where the goal is to identify and return uncertainty around the true value of $G$. Figure \ref{fig:density2} shows kernel density plots for the generated data, where the three peaks are less clear in the $n=20$ case (Figure \ref{fig:density2A}) than the $n=50$ case (Figure \ref{fig:density2B}).

\begin{figure}
\centering
   \subfloat[Density for $n=20$ observations\label{fig:density2A}]{
     \includegraphics[scale=0.85]{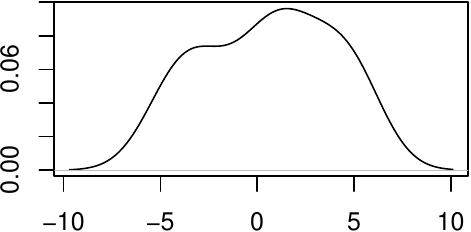}
   }
   \hfill
   \subfloat[Density for $n=50$ observations\label{fig:density2B}]{
     \includegraphics[scale=0.85]{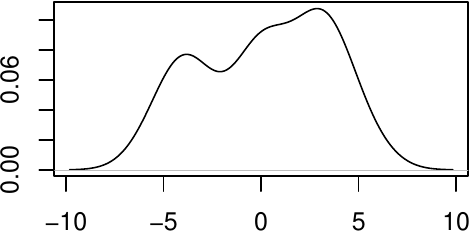}
   }
\caption{Kernel density plots for data generated from GMM with 3 components} 
\label{fig:density2}

\vspace{1cm}

\centering
   \subfloat[Components for $n=20$ observations\label{fig:dpmm2A}]{
     \includegraphics[scale=0.8]{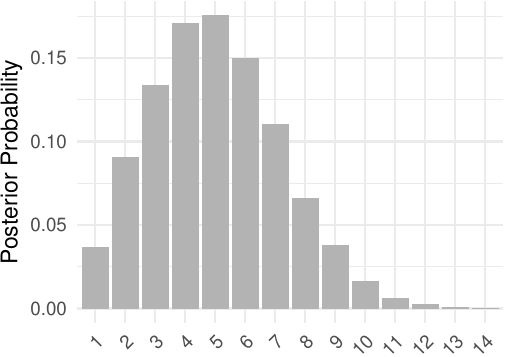}
   }
   \hfill
   \subfloat[Components for $n=50$ observations\label{fig:dpmm2B}]{
     \includegraphics[scale=0.8]{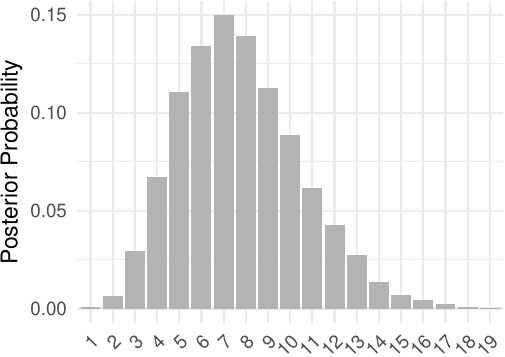}
   }
\caption{Posterior uncertainty over number of components $G$ sampled in DPMM} 
\label{fig:dpmm2}

\vspace{1cm}

\centering
   \subfloat[Components for $n=20$ observations\label{fig:densmp2A}]{
     \includegraphics[scale=0.8]{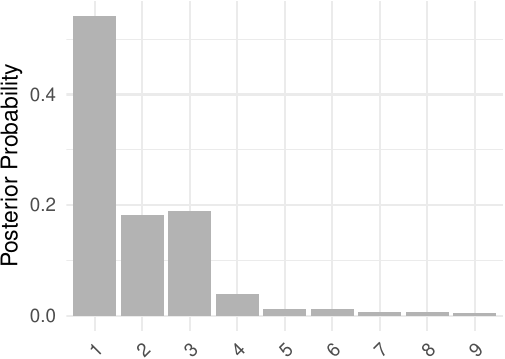}
   }
   \hfill
   \subfloat[Components for $n=50$ observations\label{fig:densmp2B}]{
     \includegraphics[scale=0.8]{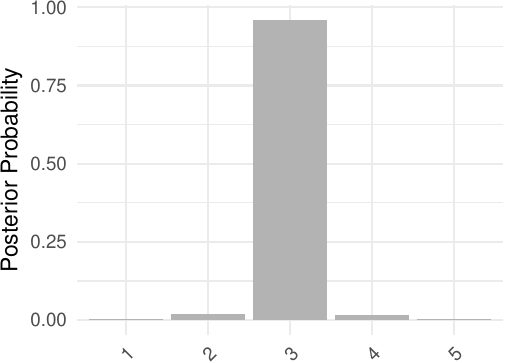}
   }
\caption{Posterior uncertainty over number of components $G$ via resampling} 
\label{fig:densmp2}
\end{figure}

Following \citet{fong_martingale_2024}, we apply Dirichlet process mixture modeling (DPMM) as our baseline for comparison. The DPMM \citep{escobar_bayesian_1995} samples from a Dirichlet process distribution $DP(H, \alpha)$ where each mixture component has its underlying parameters drawn from the base distribution $H$, and the probability of forming new components is determined by the concentration parameter $\alpha$. In practice, since directly evaluating this joint distribution is not possible, DPMM uses Gibbs sampling to return uncertainty over the number of components, their means, variances, and weights. 

We implement DPMM with the \textbf{dirichletprocess} package in R using the default priors and hyperparameters \citep{ross_dirichletprocess_2019}. For eight separate Gibbs sampling chains, we discard the first 500 iterations as burn-in and retain the next 2,000 iterations. Figure \ref{fig:dpmm2} displays the distribution of the number of components sampled in the DPMM. For the $n=20$ case (Figure \ref{fig:dpmm2A}), the mode is 5 components, and for $n=50$ (Figure \ref{fig:dpmm2B}) the mode is 7 components, both more complex than the ``true'' $G=3$. This reflects the known result that DPMM should not be used to estimate the number of components, which asymptotically tends towards infinity as $n$ increases \citep{yang_posterior_2019, cai_finite_2020}.

For the resampling approach, we implement EM clustering with the \textbf{mclust} package in R with specified candidate models ranging from 1 to 9 components. Models with both equal and unequal variances are tested, with differing dimension penalties in the BIC calculation as noted previously. We recursively simulate $N=n+600$ new observations per trial across a total of $B=400$ trials; this value of $N$ is again empirically found to be sufficiently large that $\mathcal{M}_{k(N)}$ closely approximates $\mathcal{M}_{k(\infty)}$, with convergence diagrams available in the supplementary material. 

Figure \ref{fig:densmp2} shows the distribution of the final number of components across all trials. In the $n=20$ case, the initial model with the best BIC has 1 component, and this transitions to $G=3$ in 19\% of trials, with uncertainty spread over values up to $G=9$. In the $n=50$ case, the initial model with the best BIC has 3 components, and this remains the same for 96\% of trials, with more concentrated uncertainty ranging over values from $G=1$ to 5. In both cases, the BIC penalty against complex models prevents the inflated number of components observed in the DPMM, and we retrieve reasonable posterior uncertainty over the total number of components.

\subsubsection*{Real-world example}

To demonstrate on a real-world example, we consider the galaxies dataset from \citet{roeder_density_1990}, which contains velocity measurements for $n=82$ galaxies in the Corona Borealis region. Figure \ref{fig:density3} shows a kernel density plot of the data, where the goal is to group similar galaxies by velocity. This dataset has been used for univariate clustering analysis across several prior works \citep{richardson_bayesian_1997, rodriguez_martingale_2025}.

\begin{figure}
\centering
\includegraphics[scale=0.85]{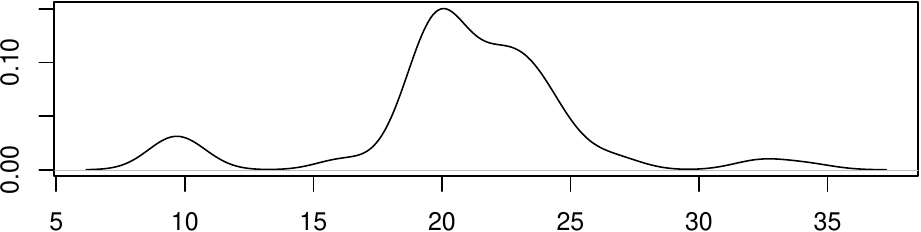}
\caption{Kernel density plot for galaxies dataset}
\label{fig:density3}

\vspace{5cm}

\centering
\includegraphics[width=\textwidth]{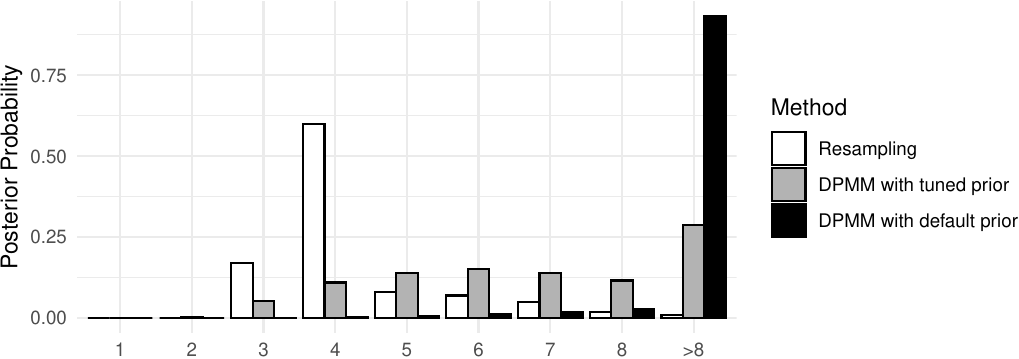}
\caption{Posterior uncertainty over number of components $G$ via resampling (white); DPMM with tuned prior $\gamma(1,8)$ for concentration parameter $\alpha$ (gray); and DPMM with default prior $\gamma(2,4)$ for concentration parameter (black)}
\label{fig:densmp3}
\end{figure}

As before, we use the \textbf{dirichletprocess} package in R, with eight separate Gibbs sampling chains containing 500 burn-in and 2,000 retained iterations. We set the prior cluster mean $\mu_0$ equal to the observed mean of 20.83 in the data, but otherwise keep the default priors and hyperparameters. In Figure \ref{fig:densmp3}, we see that this DPMM (shaded in black) again overestimates the number of components, with a mode of 15. This can be alleviated by tuning the priors -- for example, if we decrease the prior on the concentration parameter $\alpha$ from the default of $\gamma(2,4)$ to a more conservative $\gamma(1,8)$, the new mode of the DPMM (shaded in gray) becomes 6 components. However, this example reveals a key limitation of the DPMM and other sampling-based approaches, which is the underlying dependence on prior specification.

In contrast, the resampling approach returns posterior uncertainty without the need for a prior. We again use the \textbf{mclust} package in R, with candidate models ranging from 1 to 9 components using both equal and unequal variances. We recursively simulate $N=n+1500$ new observations per trial across a total of $B=100$ trials; this sample size is sufficient for the model choice to converge, with a diagram available in the supplementary material. Figure \ref{fig:densmp3} shows the distribution of the final number of components across all trials (shaded in white), with a realistic mode of 4 components.

This illustration showcases how predictive resampling can convert any general model selection technique into a probabilistic quantification of model uncertainty. We demonstrate here using density estimation with \textbf{mclust}, but the above framework would apply for any other package or method that allows the user to compare models using a consistent criterion and then sample a new observation from the best model.

\subsection{Variable selection}
\label{sec:varsel}

Another common model uncertainty question is identifying relevant variables in a high-dimensional regression. To demonstrate, we generate $n=\{10, 20, 50, 100\}$  data points across $B=100$ separate trials from a simple linear model with 20 i.i.d. normal covariates $x_1, x_2, \ldots, x_{20} \sim \mathcal{N}(0, 1)$, where only the first five are used to determine the outcome
\begin{equation}
y = x_1 + x_2 + x_3 + x_4 + x_5 + \mathcal{N}(0, 1)
\end{equation}
We expect that uncertainty on the selected variables will decrease as the observed sample size $n$ increases. In other words, we expect model choice to be diffuse across the covariates when $n=10$, and more certain when $n=100$.

As briefly discussed in Section \ref{sec:mcmethod}, within the predictive framework, we take the observed covariates $X$ as fixed, and the target at each step is to identify an optimal regression model for the prediction of new outcomes $Y$. Specifically, to maintain c.i.d. resampling, we replicate the entire design matrix $x_{1:n}$ to yield $x_{n+1:2n}$; predict $y_{n+1:2n}$ using the optimal $\bestmodel{n}$; and then append $(x,y)_{n+1:2n}$ to the observed data. 

To break this process down in more detail, if we were fitting a standard linear model without selection given $(Y_0, \ldots, Y_k)$, we would get
$$\widehat\beta_{k+1}=(X'X)^{-1}X'\left[\frac{1}{k+1}\sum_{l=0}^{k+1} Y_l\right]$$
so we take
$$Y_{k+1}=X\widehat\beta_{k+1}+\sigma\varepsilon=H\,\left[\frac{1}{k+1}\sum_{l=0}^{k+1} Y_l\right]+\sigma\,\varepsilon_{k+1}$$
and the procedure would iterate to yield a random $\widehat\beta_\infty$ from the posterior. To incorporate model selection, we compute the BIC value for each model at each iteration. We calculate
$$n(k+1)\log \left\{\frac{{\cal Y}'_k(I-H_{d,k+1}){\cal Y}_k}{n(k+1)}\right\}+d\log (n(k+1)),$$
where $d$ is the dimension of the model and $H_{k,d}$ is the $n(k+1) \times n(k+1)$ hat matrix, based on design matrix $X_d$ and comprising of blocks of $H_d$. So, for example,
$$H_{d,2}=\half\left(
\begin{array}{ll}
H_d & H_d \\
H_d & H_d
\end{array}
\right).$$
In the BIC calculation, we then have 
$${\cal Y}'_{k}(I-H_{d,k}){\cal Y}_k=\sum_{l=0}^k Y_l'Y_l-(k+1)^{-1}\sum_{0\leq l,j\leq k} Y_l'H_dY_j$$
The $1/(k+1)$ term appears throughout since
$$(X_k'X_k)^{-1}=(k+1)^{-1}(X'X)^{-1},$$
where $X_k$ is the $n(k+1) \times d$ design matrix for the data $Y_k$. At each iteration, we pick the best model $d(k)$ and then sample $Y_{k+1}$ using $H_{d(k)}$. 

We apply Gibbs sampling \citep{george_variable_1993} as our baseline for comparison, with binary inclusion indicators $\gamma_j \in \{0, 1\}$ for each covariate. The resulting Markov chain can iteratively explore different subsets of predictors by adding and removing covariates from the model specification. This approach is modeled after reversible jump Markov chain Monte Carlo (RJ-MCMC), a popular Bayesian model selection technique for exploration across models with different numbers of parameters \citep{green_reversible_1995}. However, it avoids the computational complexity of RJ-MCMC's modified Metropolis-Hastings acceptance probability, which adjusts for the change in ``volume'' of the parameter space when moving between models of different sizes.

We implement our Gibbs sampler with the \textbf{rjags} package in R \citep{plummer_rjags_2023}. For regression, we model the linear predictor term as 
\begin{equation}
\mu_i = \beta_0 + \sum_{j=1}^J (\beta_j \cdot \gamma_j \cdot X_{i, j})    
\end{equation}
where the outcomes are then $y_i \sim \mathcal{N}(\mu_i, 1/\tau)$. The prior distributions are specified as $\beta_0 = \beta_j \sim \mathcal{N}(0, 0.01)$ for the intercept term and regression coefficients; $\gamma_j \sim \text{Bernoulli}(0.5)$ for the variable inclusion indicators; and $\tau \sim \text{Gamma}(0.01, 0.01)$ for the response precision. In the context of model selection, our target parameters of interest are the variable inclusion indicators $\gamma_j$ that indicate whether each covariate is excluded or included. 

In each trial, for three separate Gibbs sampling chains, we discard the first 5,000 iterations as burn-in and retain the next 10,000 iterations. Figure \ref{fig:varRJ1} displays the mean posterior selection frequency for each covariate across the $B=100$ trials, while Figure \ref{fig:varRJ2} displays the proportion of trials for which the inclusion indicator $\gamma_j$ is greater than 0.5, indicating the variable is more likely than not to be included. In both cases, we observe a similar pattern, with significant uncertainty across all variables for $n=10$; roughly 50-50 identification of the correct variables for $n=20$; and near-certainty for $n=50$ and $n=100$.

\begin{figure}
\centering
\includegraphics[width=\textwidth]{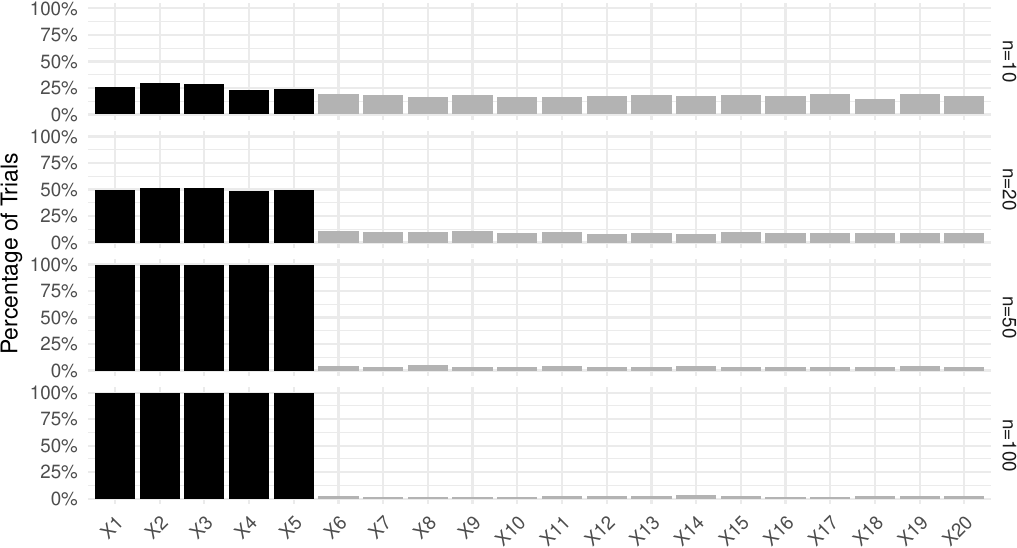}
\caption{Mean posterior selection frequencies $\sum \gamma_j / B$ for Gibbs sampling as observed sample size increases, with correct variables in black}
\label{fig:varRJ1}
\vspace{0.5cm}
\includegraphics[width=\textwidth]{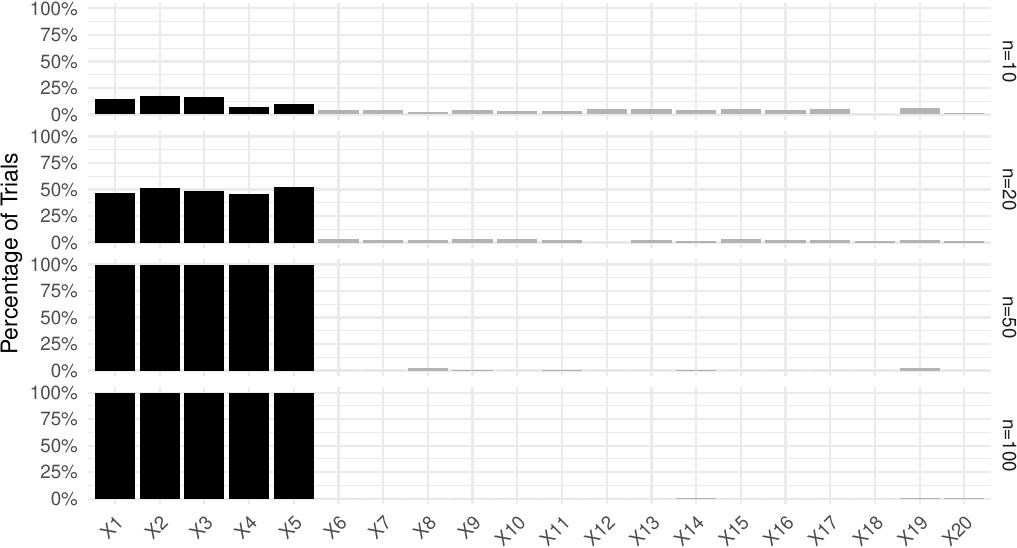}
\caption{Proportion of trials with posterior selection frequency $\gamma_j > 0.5$ for Gibbs sampling as observed sample size increases, with correct variables in black}
\label{fig:varRJ2}
\end{figure}

\begin{figure}
\centering
\includegraphics[width=\textwidth]{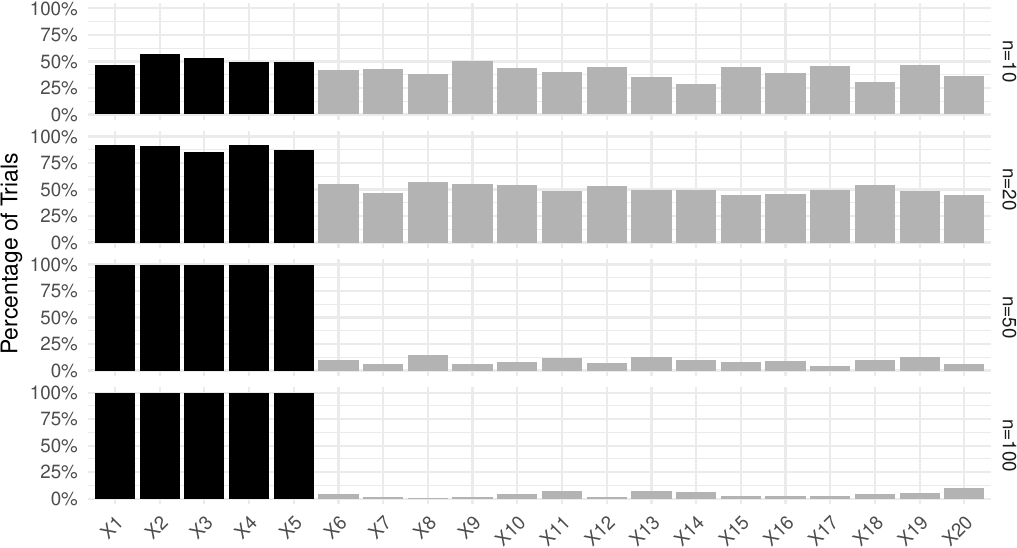}
\caption{Proportion of trials with $x_j$ in final model for forward stepwise regression with BIC as observed sample size increases, with correct variables in black}
\label{fig:varBIC}
\vspace{0.5cm}
\includegraphics[width=\textwidth]{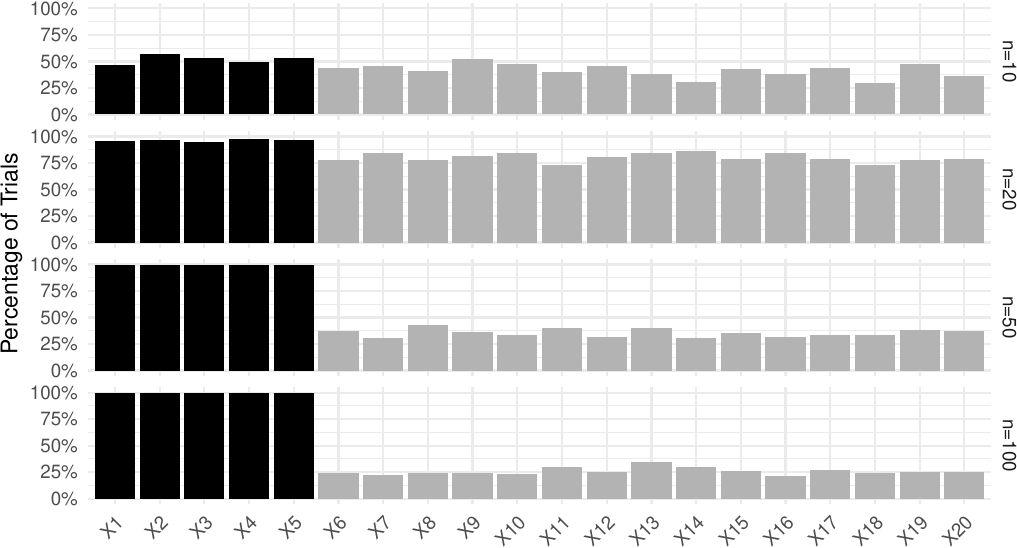}
\caption{Proportion of trials with $x_j$ in final model for forward stepwise regression with AIC as observed sample size increases, with correct variables in black}
\label{fig:varAIC}
\end{figure}

For the resampling approach, to simplify the identification of the optimal model at each step, we apply forward stepwise regression with the \textbf{MASS} package in R \citep{venables_modern_2002}. We begin with the null model (intercept only) and greedily add terms one at a time that most improve the model BIC until no further improvements can be found \citep{efroymson_multiple_1960}, then use this intermediate model to predict a new vector of outcomes $Y$. This is repeated for $b=10$ blocks of size $n$, yielding a final resampled dataset of size $n\,(b+1)$, after which we record the final selected model. We then replicate this procedure across the $B=100$ trials to return uncertainty over the distribution of final models.

Figure \ref{fig:varBIC} displays the posterior selection frequency of each variable as a function of the observed sample size, using BIC as the optimization criterion. We observe similar patterns to the Gibbs sampling approach. In the $n=10$ case, the procedure is highly uncertain and distributes weight across all available covariates, while in the $n=20$ case, the true variables clearly receive more weight. In the $n=50$ and $n=100$ cases, variables $x_1$ to $x_5$ are correctly identified in almost all trials. 

For comparison, Figure \ref{fig:varAIC} displays the posterior selection frequencies with AIC as the optimization criterion. As discussed previously, AIC is asymptotically equivalent to LOO-CV, meaning it is also asymptotically inconsistent since it tends to select overly complex models \citep{shao_linear_1993}. This pattern is reflected in the observed results, with the AIC including more variables on average across all sample sizes.

\subsection{Multi-level hypothesis testing}

For a more detailed demonstration of hypothesis testing, we consider the well-known example of the ``hot hand'' in basketball -- i.e. whether players are subject to swings in ``momentum'' that affect their game-by-game accuracy \citep{gilovich_hot_1985, kass_bayes_1995, hsiao_bayesian_2005}. 

For a particular player, we observe data $(n_i, k_i)$ for $i = 1, 2, \ldots, g$ games, where $n$ is the number of shots attempted and $k$ is the number scored. If the player's true underlying shooting percentage is stable, then the variation across individual games should be explained by binomial distributions with the same success probability. If not, the game-level shooting percentages should be sufficiently independent from each other. A simple hypothesis test that captures this difference is
\begin{align}
    H_0&: k_i \sim \text{Bin}(n_i, p) \text{ with a common percentage } p \\
    H_1&: k_i \sim \text{Bin}(n_i, p_i) \text{ with independent } p_i \sim B(a, b)
\end{align}
where the null model is binomial for all games and the alternative model is beta-binomial.

The null model likelihood is
\begin{equation}
\mathcal{L}_0(1, 1 \mid \{n_i, k_i\}_{i=1}^g) = \prod_{i=1}^g \binom{n_i}{k_i} \frac{B(k_i + 1, n_i - k_i + 1)}{B(1, 1)}
\end{equation}
with a uniform prior on the beta distribution.

For the alternative model likelihood, we have
\begin{equation}
\mathcal{L}_1(a, b \mid \{n_i, k_i\}_{i=1}^g) = \prod_{i=1}^g \binom{n_i}{k_i} \frac{B(k_i + a, n_i - k_i + b)}{B(a, b)}
\end{equation}
where the best values of $\hat{a}$ and $\hat{b}$ can be fitted with numerical optimization.

We maximize the log-likelihood for each model and compare the BIC. The number of parameters is one for the null model (the prior) and two for the alternative model ($a, b$), while the sample size is the number of games. For the model with the better BIC, we simulate a new game by bootstrapping $n_{g+1}$ from $(n_1, \ldots, n_g)$ and then drawing $k_{g+1} \sim \text{BetaBin}(n_{g+1}, 1, 1)$ in the null case, or $k_{g+1} \sim \text{BetaBin}(n_{g+1}, \hat{a}, \hat{b})$ in the alternative case.

To demonstrate, we generate data for $g=20$ games with $k_i \sim B(n_i, p_i)$, where $n_i \sim \text{Unif}[5, 15]$ and $p_i \sim B(\alpha, \alpha)$. We vary the value of $\alpha \in \{0.5, 1, 1.5, 2, 2.5\}$ and resample up to $G=200$ games across 100 trials. The figure below shows the acceptance proportion of $H_0$ with respect to the $\alpha$ used for data generation. As expected, we accept $H_0$ in most cases when the initial data are generated under $\text{BetaBin}(n_i, 1, 1)$, corresponding to the original null hypothesis, and vice versa.

\vspace{.2cm}
\begin{figure}[h]
\centering
\includegraphics[width=.7\textwidth]{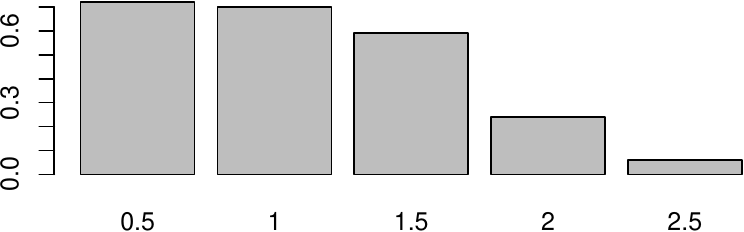}
\caption{$H_0$ is accepted in 70\% of trials when the data is generated under $B(1, 1)$, decreasing to 6\% of trials under $B(2.5, 2.5)$.}
\end{figure}


\section{Discussion}

We view model uncertainty through the lens of missing information. With a consistent model selection criterion in hand, we would be able to reliably identify the correct model if we had the complete data. The imputation of new observations allows us to convert such a model selection criterion directly into probabilities over the space of candidate models, by propagating uncertainty through different possible realizations of the missing data. Predictive resampling \citep{fong_martingale_2024} serves as our mechanism for imputation, and we see that model choice under this framework converges as the sample size grows. Our approach serves as a form of model expansion around the initial best model for the observed data, as discussed by \citet{draper_assessment_1995}, and also echoes the prequential argument of \citet{dawid_present_1984} with its focus on step-by-step prediction as the fundamental object of statistical modeling. This framework can be applied to any general method or package, as long as it allows for models to be compared and for a new observation to then be sampled from the best model.

Computationally, we avoid certain complexities entailed by Bayes factor calculations and MCMC-style sampling methods, with the additional benefit of bypassing the need for subjective prior specification. We acknowledge challenges related to Monte Carlo uncertainty quantification. Specifically, our method relies on indexing uncertainty over several resampling repetitions to obtain precise results, which becomes increasingly difficult as the model space grows. The ability to parallelize the resampling process is a valuable asset that can alleviate some of these computational demands. Another question that warrants attention is the requirement for models to be easily updated and for the consistent model selection criterion to be rapidly optimized at each step. While we have not specifically explored this aspect in the current work, it presents a potential area for future research. Efficient iteration of the model search process -- for example, through an online search process that maintains a working model and uses a gradient-based update at each step -- will be critical for scaling our method to more complex or higher-dimensional model spaces.

\subsubsection*{Acknowledgements}

Shirvaikar is supported by the EPSRC Centre for Doctoral Training in Modern Statistics and Statistical Machine Learning (EP/S023151/1) and Novo Nordisk. Holmes is supported by Novo Nordisk. 

\bibliography{export-data}

\newpage

\section{Supplementary material}

\subsection{Convergence diagrams for two-sided hypothesis testing}

The first set of diagrams is for the demonstration in Section \ref{sec:hyptest}, where we test $H_0: \theta = \theta_0$ against $H_1: \theta \neq \theta_0$ for a normal mean with $\sigma^2 = 1$. For $n = \{30, 100, 300, 1000\}$ observations simulated from the ``true alternative'' $\mathcal{N}(0.1,1)$, we see in Figures \ref{fig:jellyH1_30}, \ref{fig:jellyH1_100}, \ref{fig:jellyH1_300}, and \ref{fig:jellyH1_1000} respectively that $N = n + 20n$ additional observations are sufficient for $\bestmodel{N}$ to closely approximate the final $\mathcal{M}_{k(\infty)}$ over 100 random trials.

We note that the scale of the y-axis (the range of final model means) decreases as the initial observed $n$ increases, reflecting the intuition that uncertainty should be reduced with the observation of additional data.

\subsection{Comparison plots for two-sided hypothesis testing}

In Figure \ref{fig:h0demoE}, we plot the (log-scaled) $e$-values for each seed on the horizontal axis and the resampling posterior probabilities of $H_0$ on the vertical axis, for data generated under $H_0$. This corresponds to Figure \ref{fig:h0demo} for $p$-values. The X marks on the plots indicate tests with $e>10$ where Jeffreys' rule of thumb finds strong evidence against $H_0$, and the O marks indicate $e<10$. We note that the $e$-values are invariant to sample size, with a constant Type I error rate regardless of $n$, but that the resampling probabilities tend towards 1 as the sample size grows.

Figure \ref{fig:h1demoE} is the corresponding diagram to Figure \ref{fig:h1demo} for $p$-values, with data generated under $H_1$. The $e$-values build evidence against the null as $n$ increases, with the points gradually migrating towards the bottom and right as we observe more data.

\subsection{Convergence diagrams for density estimation}

The next set of diagrams is for the illustrations in Section \ref{sec:density}. For the first example with two components, we track the ongoing model choice between $G=1$ and $G=2$ components as additional data points are imputed across 100 random trials. We add some jitter to each individual resampling line to improve the visibility of the diagram. In Figures \ref{fig:jellyEM1_06} and \ref{fig:jellyEM1_09}, for data with $\sigma=0.6$ and $\sigma=0.9$ respectively, $N=n+200$ additional observations are sufficient for the choice of model to converge.

For the second example with three components, we now plot candidate models with $G = \{1, \ldots, 9\}$ components across 100 random trials. We again add some jitter to each individual resampling line for visibility. In Figure \ref{fig:jellyEM2_20}, we see that models with up to 9 components are explored, and that $N=n+600$ additional observations are sufficient for $\bestmodel{N}$ to approximate $\mathcal{M}_{k(\infty)}$. In Figure \ref{fig:jellyEM2_50}, we see that model space is only explored up to $G=4$ components, reflecting the principle that observing more initial data should tighten our final uncertainty estimate.

Finally, for the real-world galaxies dataset, we again plot candidate models with $G = \{1, \ldots, 9\}$ components across 100 random trials, with some jitter for visibility. Figure \ref{fig:jellyEM3} shows that $N=n+1500$ additional observations are sufficient for the choice of model to converge.

\newpage 

\begin{figure}
\centering
\includegraphics[width=\textwidth]{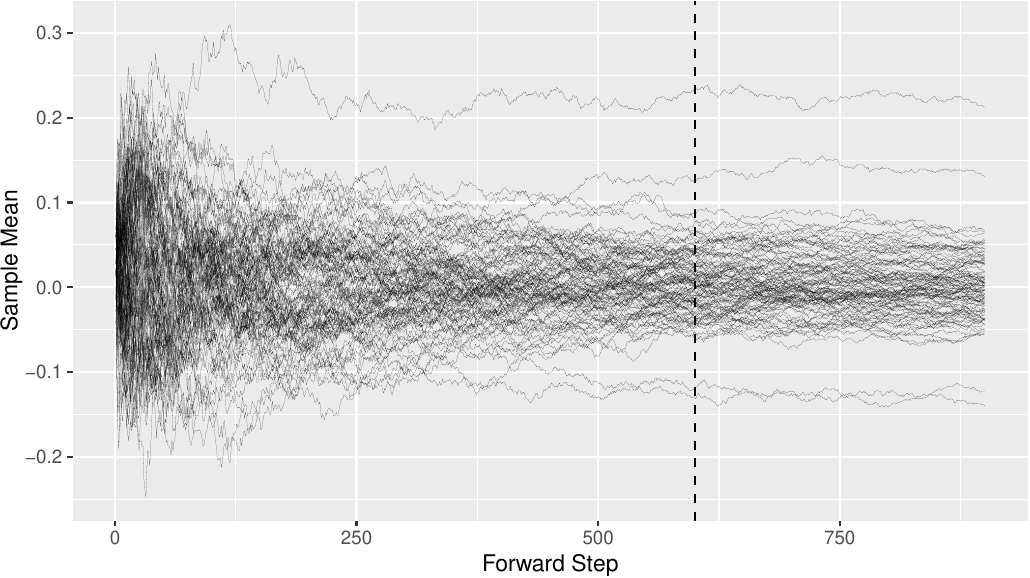}
\vspace{-0.5cm}
\caption{Example convergence diagram for predictive resampling with $m=30$ observations simulated from $\mathcal{N}(0.1,1)$, showing that model choice converges after $M=m+20m$ imputed observations, indicated by the dashed vertical line.}
\label{fig:jellyH1_30}
\includegraphics[width=\textwidth]{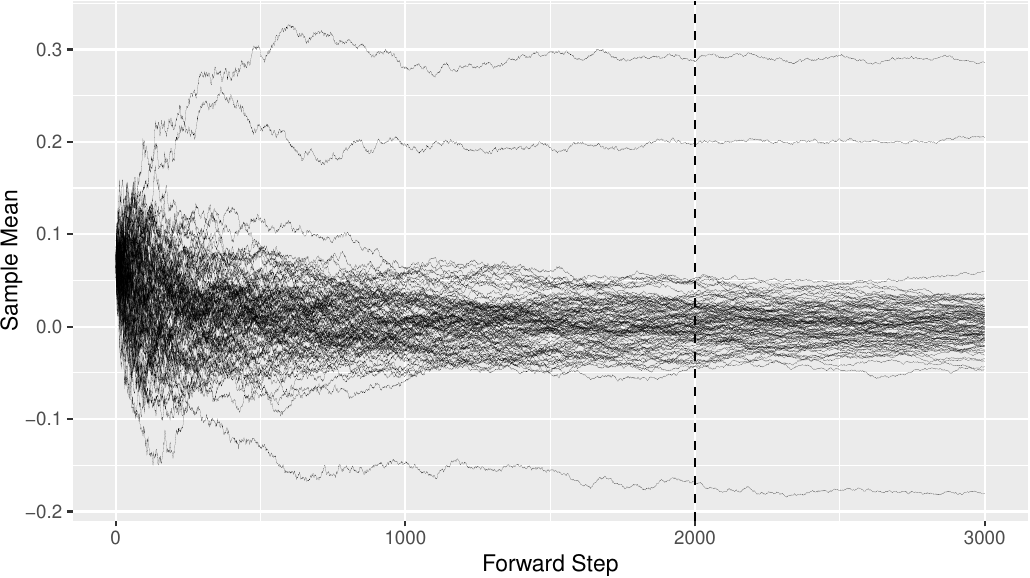}
\vspace{-0.5cm}
\caption{Example convergence diagram for predictive resampling with $m=100$ observations simulated from $\mathcal{N}(0.1,1)$, showing that model choice converges after $M=m+20m$ imputed observations, indicated by the dashed vertical line.}
\label{fig:jellyH1_100}
\end{figure}

\begin{figure}
\centering
\includegraphics[width=\textwidth]{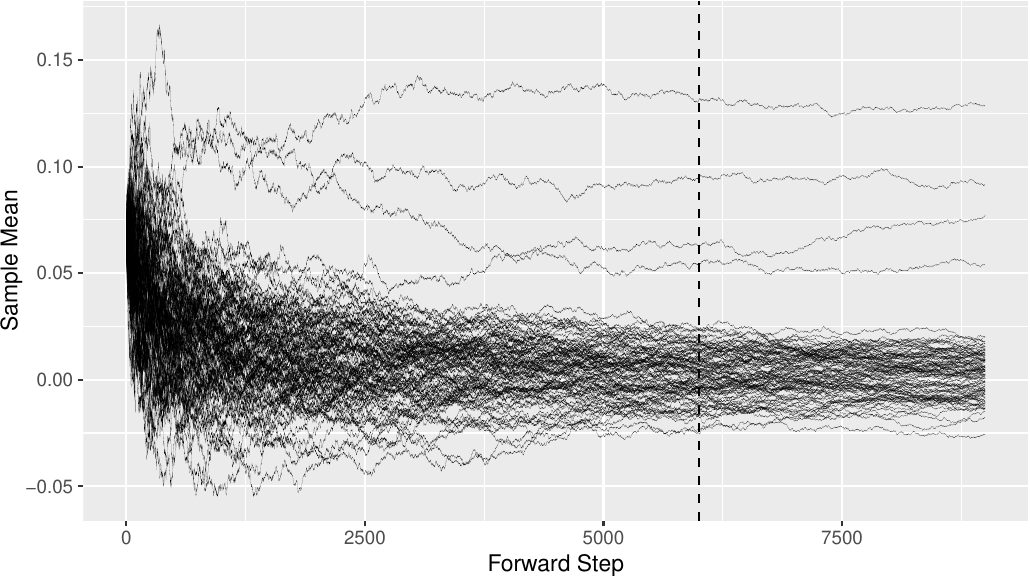}
\vspace{-0.5cm}
\caption{Example convergence diagram for predictive resampling with $m=300$ observations simulated from $\mathcal{N}(0.1,1)$, showing that model choice converges after $M=m+20m$ imputed observations, indicated by the dashed vertical line.}
\label{fig:jellyH1_300}
\includegraphics[width=\textwidth]{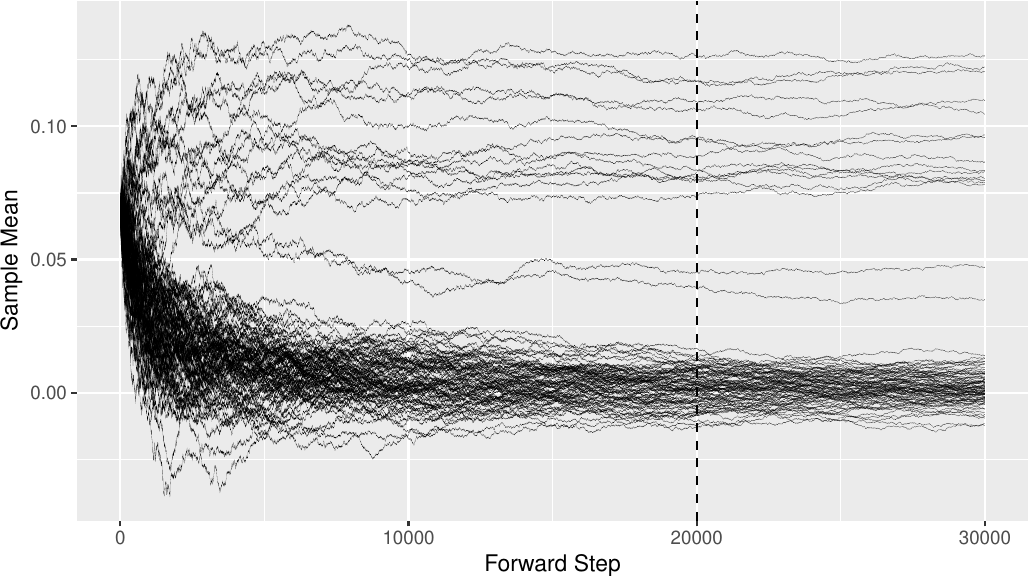}
\vspace{-0.5cm}
\caption{Example convergence diagram for predictive resampling with $m=1000$ observations simulated from $\mathcal{N}(0.1,1)$, showing that model choice converges after $M=m+20m$ imputed observations, indicated by the dashed vertical line.}
\label{fig:jellyH1_1000}
\end{figure}

\begin{figure}
\centering
\includegraphics[width=\textwidth]{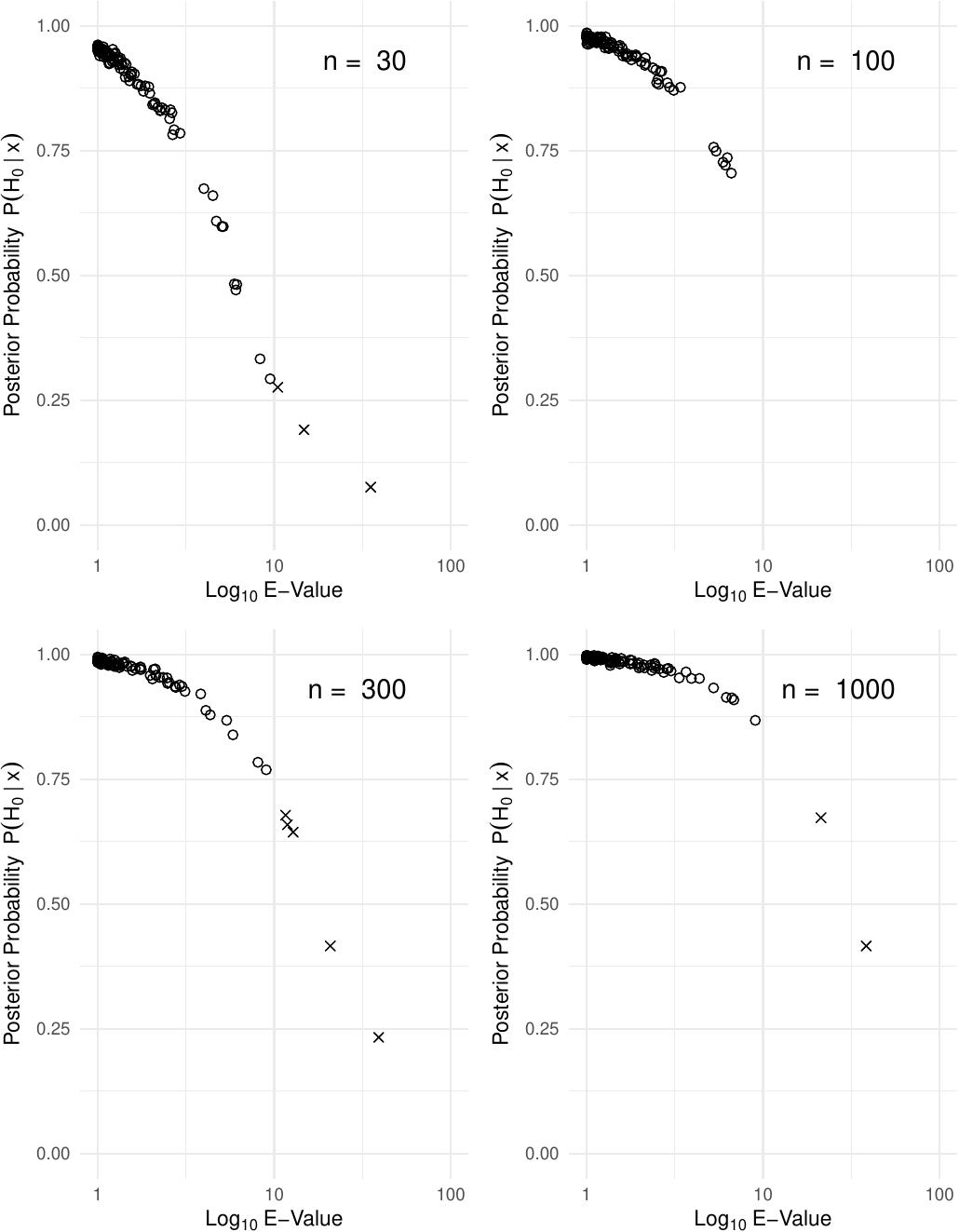}
\vspace{-0.5cm}
\caption{Observed $e$-value (log-scale) vs. resampling posterior probability of $H_0$ for data generated under the null $\mathcal{N}(0, 1)$ across 100 random seeds. X denotes tests with $e>10$ where Jeffreys' rule of thumb finds strong evidence against $H_0$, and O denotes tests with $e<10$.}
\label{fig:h0demoE}
\end{figure}

\begin{figure}
\centering
\includegraphics[width=\textwidth]{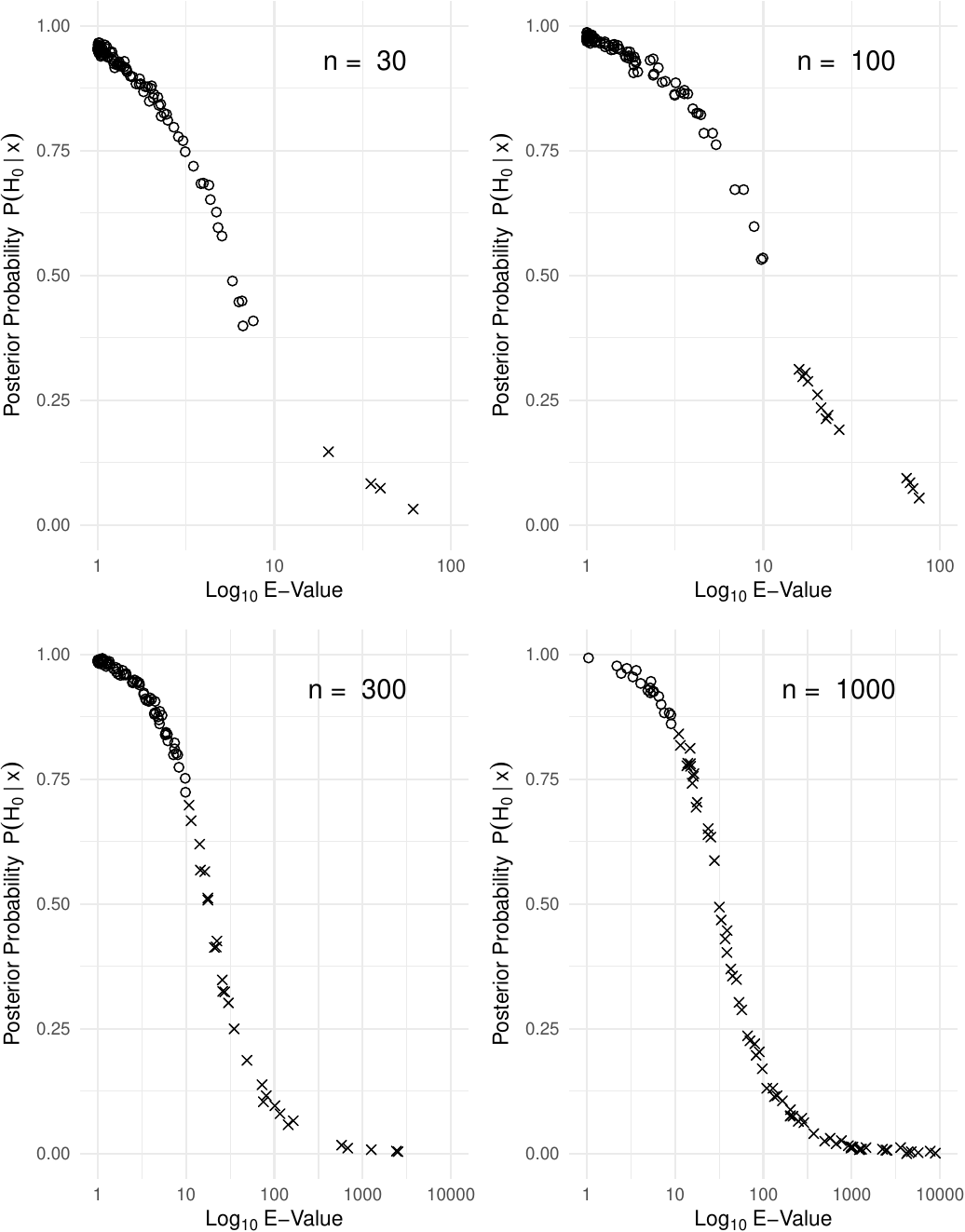}
\vspace{-0.5cm}
\caption{Observed $e$-value (log-scale) vs. resampling posterior probability of $H_0$ for data generated under the alternative $\mathcal{N}(0.1, 1)$ across 100 random seeds. X denotes tests with $e>10$ where Jeffreys' rule of thumb finds strong evidence against $H_0$, and O denotes tests with $e<10$.}
\label{fig:h1demoE}
\end{figure}

\begin{figure}
\centering
\includegraphics[width=\textwidth]{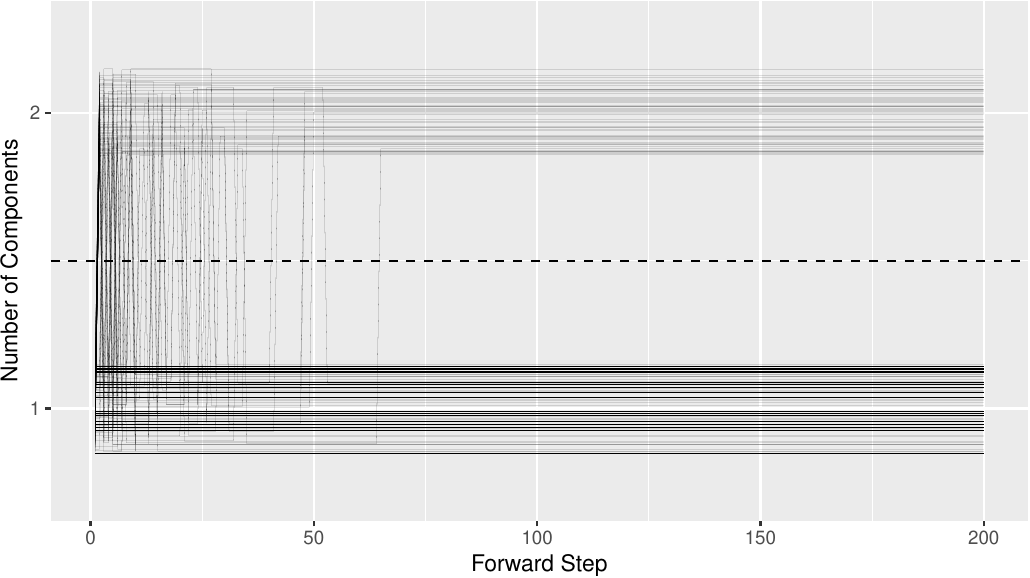}
\vspace{-0.5cm}
\caption{Example convergence diagram for density estimation in two-component GMM with $\sigma=0.6$, showing that model choice converges after $N=n+200$ imputed observations.}
\label{fig:jellyEM1_06}
\includegraphics[width=\textwidth]{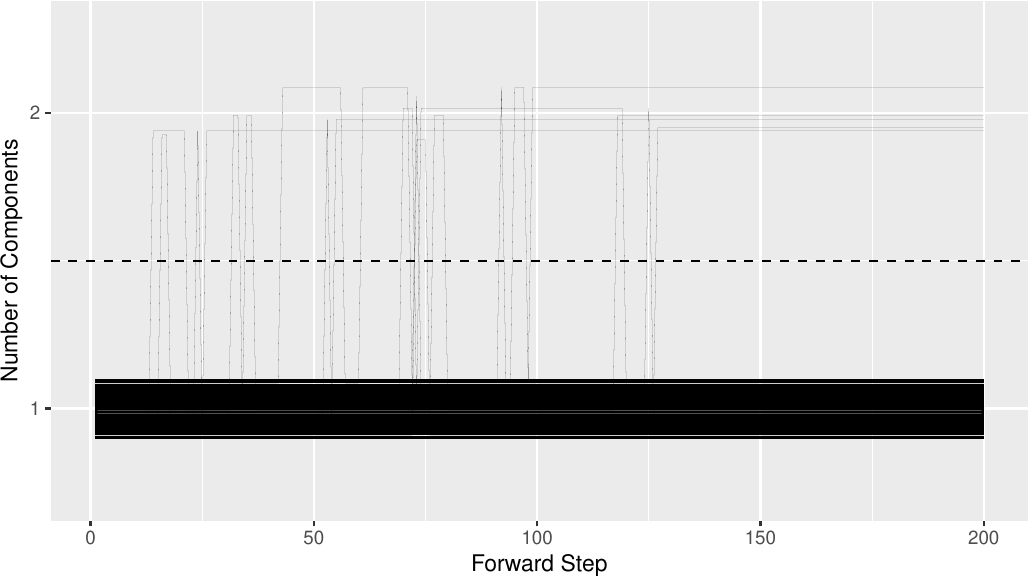}
\vspace{-0.5cm}
\caption{Example convergence diagram for density estimation in two-component GMM with $\sigma=0.9$, showing that model choice converges after $N=n+200$ imputed observations.}
\label{fig:jellyEM1_09}
\end{figure}

\begin{figure}
\centering
\includegraphics[width=\textwidth]{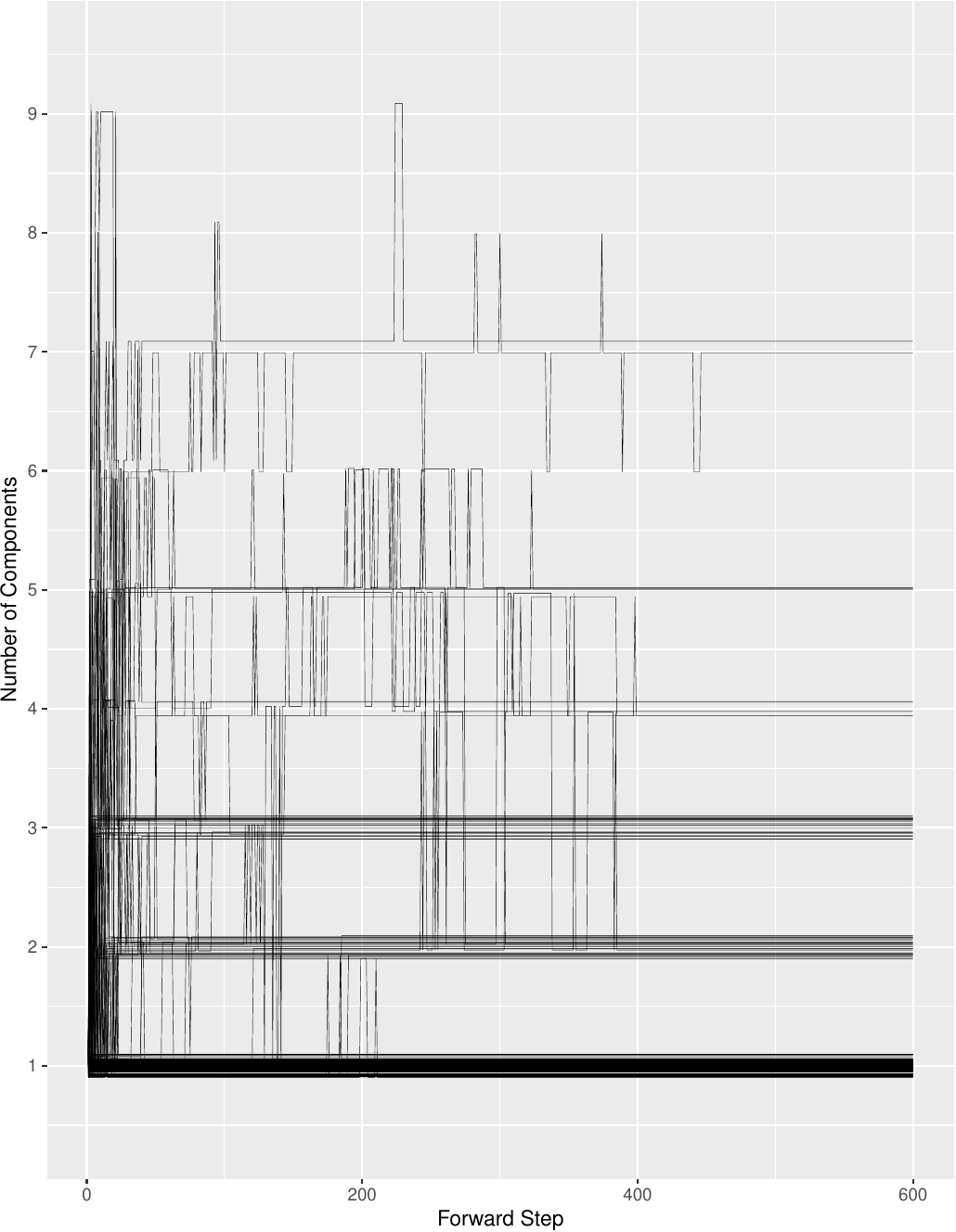}
\vspace{-0.5cm}
\caption{Example convergence diagram for density estimation in three-component GMM with $n=20$ initial data points, showing that model choice converges after $N=n+600$ imputed observations.}
\label{fig:jellyEM2_20}
\end{figure}

\begin{figure}
\centering
\includegraphics[width=\textwidth]{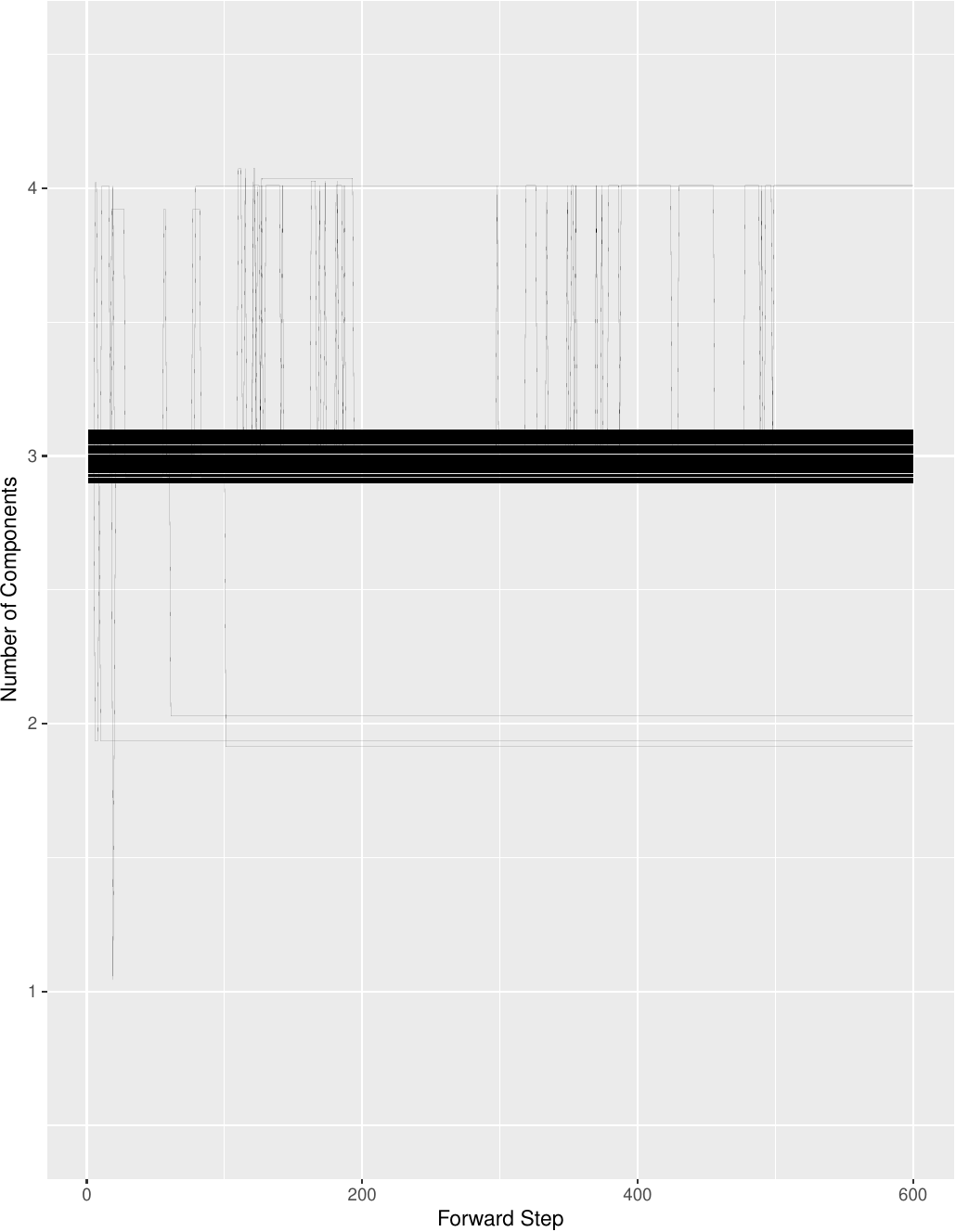}
\vspace{-0.5cm}
\caption{Example convergence diagram for density estimation in three-component GMM with $n=50$ initial data points, showing that model choice converges after $N=n+600$ imputed observations.}
\label{fig:jellyEM2_50}
\end{figure}

\begin{figure}
\centering
\includegraphics[width=\textwidth]{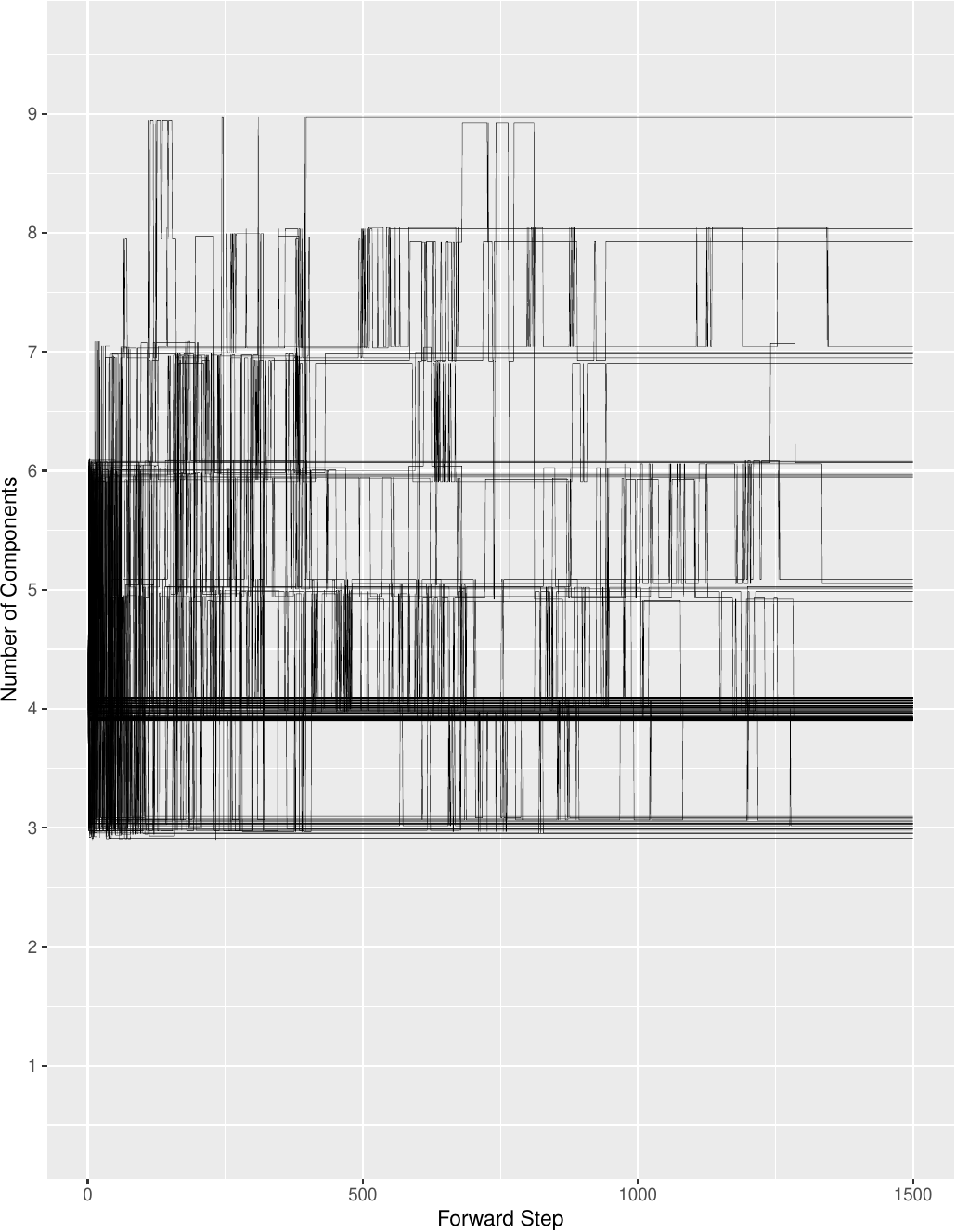}
\caption{Example convergence diagram for density estimation with galaxies dataset, showing that model choice converges after $N=n+1500$ imputed observations.}
\label{fig:jellyEM3}
\end{figure}

\end{document}